\documentclass[a4paper,fullpage,11pt]{article}

\usepackage{microtype}
\usepackage{graphicx,mathdots,mathtools,amsthm,hyperref}
\usepackage{tikzfig}
\usepackage{stmaryrd}
\usepackage{amssymb,nicefrac}

\tikzstyle{env}=[copoint,regular polygon rotate=0,minimum width=0.2cm, fill=black]

\tikzstyle{every picture}=[baseline=-0.25em]
\tikzstyle{dotpic}=[scale=0.5]
\tikzstyle{diredges}=[every to/.style={diredge}]
\tikzstyle{dot graph}=[shorten <=-0.1mm,shorten >=-0.1mm,scale=0.6]
\tikzstyle{plot point}=[circle,fill=black,minimum width=2mm,inner sep=0]

\tikzstyle{braceedge}=[decorate,decoration={brace,amplitude=2mm,raise=-1mm}]
\tikzstyle{small braceedge}=[decorate,decoration={brace,amplitude=1mm,raise=-1mm}]
\tikzstyle{left hook arrow}=[left hook-latex]
\tikzstyle{right hook arrow}=[right hook-latex]

\tikzstyle{black dot}=[inner sep=0.7mm,minimum width=0pt,minimum height=0pt,fill=black,draw=black,shape=circle]

\tikzstyle{dot}=[black dot]
\tikzstyle{smalldot}=[inner sep=0.4mm,minimum width=0pt,minimum height=0pt,fill=black,draw=black,shape=circle]
\tikzstyle{white dot}=[dot,fill=white]
\tikzstyle{antipode}=[white dot,inner sep=0.3mm,font=\footnotesize]
\tikzstyle{smallwhitedot}=[smalldot,fill=white]
\tikzstyle{alt white dot}=[white dot,label={[xshift=3.07mm,yshift=-0.05mm,font=\footnotesize]left:$*$}]
\tikzstyle{gray dot}=[dot,fill=gray!40!white]
\tikzstyle{smallgraydot}=[smalldot,fill=gray!40!white]
\tikzstyle{box vertex}=[draw=black,rectangle]
\tikzstyle{small box}=[box vertex,fill=white]
\tikzstyle{whitebg}=[fill=white,inner sep=2pt]
\tikzstyle{graph state vertex}=[sg vertex,fill=black]

\tikzstyle{wide copoint}=[fill=white,draw=black,shape=isosceles triangle,shape border rotate=90,isosceles triangle stretches=true,inner sep=1pt,minimum width=1.5cm,minimum height=5mm]
\tikzstyle{wide point}=[fill=white,draw=black,shape=isosceles triangle,shape border rotate=-90,isosceles triangle stretches=true,inner sep=1pt,minimum width=1.5cm,minimum height=4mm]
\tikzstyle{very wide copoint}=[fill=white,draw=black,shape=isosceles triangle,shape border rotate=-90,isosceles triangle stretches=true,inner sep=1pt,minimum width=2.5cm,minimum height=4mm]
\tikzstyle{very wide empty copoint}=[draw=black,shape=isosceles triangle,shape border rotate=-90,isosceles triangle stretches=true,inner sep=1pt,minimum width=2.5cm,minimum height=4mm]
\tikzstyle{symm}=[ultra thick,shorten <=-1mm,shorten >=-1mm]

\tikzstyle{square box}=[rectangle,fill=white,draw=black,minimum height=5mm,minimum width=5mm,font=\small]
\tikzstyle{square gray box}=[rectangle,fill=gray!30,draw=black,minimum height=6mm,minimum width=6mm]
\tikzstyle{copoint}=[regular polygon,regular polygon sides=3,draw=black,scale=0.75,inner sep=-0.5pt,minimum width=7mm,fill=white]
\tikzstyle{point}=[regular polygon,regular polygon sides=3,draw=black,scale=0.75,inner sep=-0.5pt,minimum width=7mm,fill=white,regular polygon rotate=180]
\tikzstyle{gray point}=[point,fill=gray!40!white]
\tikzstyle{gray copoint}=[copoint,fill=gray!40!white]

\newcommand{\edgearrow}{{\arrow[black]{>}}}
\newcommand{\edgetick}{{\arrow[black,scale=0.7,very thick]{|}}}

\tikzstyle{diredge}=[->]
\tikzstyle{rdiredge}=[<-]
\tikzstyle{medium diredge}=[->]

\tikzstyle{short diredge}=[->]
\tikzstyle{halfedge}=[-)]
\tikzstyle{other halfedge}=[(-]
\tikzstyle{freeedge}=[(-)]
\tikzstyle{white edge}=[line width=5pt,white]
\tikzstyle{tick}=[postaction=decorate,decoration={markings, mark=at position 0.5 with \edgetick}]
\tikzstyle{small map edge}=[|-latex, gray!60!blue, shorten <=0.9mm, shorten >=0.5mm]
\tikzstyle{thick dashed edge}=[very thick,dashed,gray!40]
\tikzstyle{map edge}=[|-latex,very thick, gray!40, shorten <=1mm, shorten >=0.5mm]
\tikzstyle{tickedge}=[postaction=decorate,
  decoration={markings, mark=at position 0.5 with \edgetick}]
\tikzstyle{dirtickedge}=[postaction=decorate,
  decoration={markings, mark=at position 0.5 with \edgetick},
  decoration={markings, mark=at position 0.85 with \edgearrow}]
\tikzstyle{dirdoubletickedge}=[postaction=decorate,
  decoration={markings, mark=at position 0.4 with \edgetick},
  decoration={markings, mark=at position 0.6 with \edgetick},
  decoration={markings, mark=at position 0.85 with \edgearrow}]

\makeatletter
\newcommand{\boxshape}[3]{%
\pgfdeclareshape{#1}{
\inheritsavedanchors[from=rectangle] 
\inheritanchorborder[from=rectangle]
\inheritanchor[from=rectangle]{center}
\inheritanchor[from=rectangle]{north}
\inheritanchor[from=rectangle]{south}
\inheritanchor[from=rectangle]{west}
\inheritanchor[from=rectangle]{east}
\backgroundpath{
\southwest \pgf@xa=\pgf@x \pgf@ya=\pgf@y
\northeast \pgf@xb=\pgf@x \pgf@yb=\pgf@y

\@tempdima=#2
\@tempdimb=#3

\pgfpathmoveto{\pgfpoint{\pgf@xa - 5pt + \@tempdima}{\pgf@ya}}
\pgfpathlineto{\pgfpoint{\pgf@xa - 5pt - \@tempdima}{\pgf@yb}}
\pgfpathlineto{\pgfpoint{\pgf@xb + 5pt + \@tempdimb}{\pgf@yb}}
\pgfpathlineto{\pgfpoint{\pgf@xb + 5pt - \@tempdimb}{\pgf@ya}}
\pgfpathlineto{\pgfpoint{\pgf@xa - 5pt + \@tempdima}{\pgf@ya}}
\pgfpathclose
}
}}

\boxshape{NEbox}{0pt}{8pt}
\boxshape{SEbox}{0pt}{-8pt}
\boxshape{NWbox}{8pt}{0pt}
\boxshape{SWbox}{-8pt}{0pt}
\makeatother

\tikzstyle{map}=[draw,shape=NEbox,inner sep=7pt]
\tikzstyle{mapdag}=[draw,shape=SEbox,inner sep=7pt]
\tikzstyle{maptrans}=[draw,shape=SWbox,inner sep=7pt]
\tikzstyle{mapconj}=[draw,shape=NWbox,inner sep=7pt]

\tikzstyle{probs}=[shape=semicircle,fill=gray!40!white,draw=black,shape border rotate=180,minimum width=1.2cm]

\tikzstyle{arrs}=[-latex,font=\small,auto]
\tikzstyle{arrow plain}=[arrs]
\tikzstyle{arrow dashed}=[dashed,arrs]
\tikzstyle{arrow bold}=[very thick,arrs]
\tikzstyle{arrow hide}=[draw=white!0,-]
\tikzstyle{arrow reverse}=[latex-]
\tikzstyle{cdnode}=[]

\tikzstyle{gn}=[dot,fill=green,minimum width=0.3cm,inner sep=0pt]
\tikzstyle{rn}=[dot,fill=red,inner sep=0pt,minimum width=0.3cm]
\tikzstyle{bn}=[dot,fill=blue,minimum width=0.3cm]

\tikzstyle{sgn}=[dot,fill=green,minimum width=0.2cm,inner sep=0pt]
\tikzstyle{srn}=[dot,fill=red,inner sep=0pt,minimum width=0.2cm]
\tikzstyle{sbn}=[dot,fill=blue,minimum width=0.3cm]

\tikzstyle{rc}=[dot,thick,fill=white,draw = red,minimum width=0.3cm,inner sep=0pt]
\tikzstyle{gc}=[dot,thick,fill=white,draw= green,inner sep=0pt,minimum width=0.3cm]
\tikzstyle{bc}=[dot,thick,fill=white,draw= blue,minimum width=0.3cm]

\tikzstyle{label}=[circle,fill=white,minimum width=0.3cm]

\tikzstyle{H box}=[rectangle,fill=yellow,draw=black,xscale=1,yscale=1,font=\small,inner sep=1.95pt]
\tikzstyle{clocklabel}=[dot,fill=yellow,draw=black,font=\tiny,inner sep=0.75pt]

\tikzstyle{rsn}=[circle split,draw,fill=red,font=\tiny,inner sep=0.75pt]
\tikzstyle{gsn}=[circle split,draw,fill=green,font=\tiny,inner sep=0.75pt]
\tikzstyle{bsn}=[circle split,draw,fill=blue,font=\tiny,inner sep=0.75pt]

\tikzstyle{rsc}=[circle split,thick,draw= red,draw,fill=white,font=\tiny,inner sep=0.75pt]
\tikzstyle{gsc}=[circle split,thick,draw= green,draw,fill=white,font=\tiny,inner sep=0.75pt]
\tikzstyle{bsc}=[circle split,thick,draw= blue,draw,fill=white,font=\tiny,inner sep=0.75pt]

\tikzstyle{cnot}=[fill=white,shape=circle,inner sep=-1.4pt]
\tikzstyle{wire label}=[font=\tiny, auto]

\newcommand{\ket}[1]{\ensuremath{\left|  #1 \right\rangle}}

\newcommand{\denote}[1]{
\llbracket #1 \rrbracket} 

\newcommand{\denoteb}[1]{
\left\llbracket #1 \right\rrbracket}

\tikzstyle{cdiag}=[matrix of math nodes, row sep=3em, column sep=3em, text height=1.5ex, text depth=0.25ex,inner sep=0.5em]
\tikzstyle{arrow above}=[transform canvas={yshift=0.5ex}]
\tikzstyle{arrow below}=[transform canvas={yshift=-0.5ex}]

\usepackage[cm]{fullpage}
\usepackage{layout}

\newtheorem{lemma}{Lemma}
\newtheorem{definition}{Definition}
\newtheorem{theorem}{Theorem}
\newtheorem{corollary}{Corollary}
\newtheorem{remark}{Remark}

\tikzstyle{square green box}=[rectangle,fill=green!30,draw=black]
\tikzstyle{square red box}=[rectangle,fill=red!30,draw=black]
\newcommand{\ed}{\end{document}}

\newcommand{\dH}{%
\beginpgfgraphicnamed{scalars//Had4}
\InputIfFileExists{scalars//Had4.tikz}{}{\input{./figures/scalars//Had4.tikz}}
\endpgfgraphicnamed}
\newcommand{\sdH}{%
\beginpgfgraphicnamed{scalars//Had5}
\InputIfFileExists{scalars//Had5.tikz}{}{\input{./figures/scalars//Had5.tikz}}
\endpgfgraphicnamed}
\newcommand{\dempty}{%
\beginpgfgraphicnamed{scalars//emptysquare-small}
\InputIfFileExists{scalars//emptysquare-small.tikz}{}{\input{./figures/scalars//emptysquare-small.tikz}}
\endpgfgraphicnamed}

\newcommand{\ldsigma}{%
\beginpgfgraphicnamed{scalars//swap-large}
\InputIfFileExists{scalars//swap-large.tikz}{}{\input{./figures/scalars//swap-large.tikz}}
\endpgfgraphicnamed}

\newcommand{\drcup}{\raisebox{-0.15cm}{%
\beginpgfgraphicnamed{scalars//cup}
\InputIfFileExists{scalars//cup.tikz}{}{\input{./figures/scalars//cup.tikz}}
\endpgfgraphicnamed}}
\newcommand{\drcap}{\raisebox{0.15cm}{%
\beginpgfgraphicnamed{scalars//cap}
\InputIfFileExists{scalars//cap.tikz}{}{\input{./figures/scalars//cap.tikz}}
\endpgfgraphicnamed}}

\newcommand{\dsigmaf}{%
\beginpgfgraphicnamed{scalars//swapf}
\InputIfFileExists{scalars//swapf.tikz}{}{\input{./figures/scalars//swapf.tikz}}
\endpgfgraphicnamed}
\newcommand{\dHf}{%
\beginpgfgraphicnamed{scalars//Had4f}
\begin{tikzpicture}
	\begin{pgfonlayer}{nodelayer}
		\node [style={H box}] (0) at (0.5, 0) {};
		\node [style=none] (1) at (0.5, 0.3) {};
		\node [style=none] (2) at (0.5, -0.3) {};
	\end{pgfonlayer}
	\begin{pgfonlayer}{edgelayer}
		\draw (1.center) to (0);
		\draw (2.center) to (0);
	\end{pgfonlayer}
\end{tikzpicture}}
\endpgfgraphicnamed}
\newcommand{\drcupf}{\raisebox{-0.15cm}{%
\beginpgfgraphicnamed{scalars//cupf}
\InputIfFileExists{scalars//cupf.tikz}{}{\input{./figures/scalars//cupf.tikz}}
\endpgfgraphicnamed}}
\newcommand{\drcapf}{\raisebox{0.15cm}{%
\beginpgfgraphicnamed{scalars//capf}
\InputIfFileExists{scalars//capf.tikz}{}{\input{./figures/scalars//capf.tikz}}
\endpgfgraphicnamed}}

\newcommand{\dgd}{\!%
\beginpgfgraphicnamed{scalars//sspiderg}
\InputIfFileExists{scalars//sspiderg.tikz}{}{\input{./figures/scalars//sspiderg.tikz}}
\endpgfgraphicnamed\!}
\newcommand{\drd}{\!%
\beginpgfgraphicnamed{scalars//sspiderr}
\InputIfFileExists{scalars//sspiderr.tikz}{}{\input{./figures/scalars//sspiderr.tikz}}
\endpgfgraphicnamed\!}
\newcommand{\droot}{\raisebox{-0.1cm}{%
\beginpgfgraphicnamed{scalars//1line}
\begin{tikzpicture}[scale=0.55]
	\begin{pgfonlayer}{nodelayer}
		\node [style=srn] (0) at (0, 0.5) {};
		\node [style=sgn] (1) at (0, -0) {};
	\end{pgfonlayer}
	\begin{pgfonlayer}{edgelayer}
		\draw (0) to (1);
	\end{pgfonlayer}
\end{tikzpicture}}
\endpgfgraphicnamed}}
\newcommand{\diroot}{\raisebox{-0.1cm}{%
\beginpgfgraphicnamed{scalars//iroot}
\begin{tikzpicture}[scale=0.55]
	\begin{pgfonlayer}{nodelayer}
		\node [style=srn] (0) at (0, 0.5) {};
		\node [style=sgn] (1) at (0, -0) {};
	\end{pgfonlayer}
	\begin{pgfonlayer}{edgelayer}
		\draw (0) to (1);
		\draw[bend right=45, looseness=1.00] (0) to (1);
		\draw[bend left=45, looseness=1.00] (0) to (1);
	\end{pgfonlayer}
\end{tikzpicture}}
\endpgfgraphicnamed}}

\newcommand{\did}{%
\beginpgfgraphicnamed{scalars//Id}
\begin{tikzpicture}
	\begin{pgfonlayer}{nodelayer}
		\node [style=none] (1) at (0.5, 0.3) {};
		\node [style=none] (2) at (0.5, -0.3) {};
		\node [style=none] (3) at (0.5, -0.5) {};
		\node [style=none] (4) at (0.5, 0.5) {};
	\end{pgfonlayer}
	\begin{pgfonlayer}{edgelayer}
		\draw (1.center) to (2.center);
	\end{pgfonlayer}
\end{tikzpicture}}
\endpgfgraphicnamed}
\newcommand{\didf}{%
\beginpgfgraphicnamed{scalars//Idf}
\begin{tikzpicture}
	\begin{pgfonlayer}{nodelayer}
		\node [style=none] (1) at (0.5, 0.3) {};
		\node [style=none] (2) at (0.5, -0.3) {};
	\end{pgfonlayer}
	\begin{pgfonlayer}{edgelayer}
		\draw (1.center) to (2.center);
	\end{pgfonlayer}
\end{tikzpicture}}
\endpgfgraphicnamed}

\newcommand{\gdzz}{%
\beginpgfgraphicnamed{scalars//RZ00alpha}
\begin{tikzpicture}
	\begin{pgfonlayer}{nodelayer}
		\node [style=gn] (0) at (0.5, 0) {\footnotesize$\alpha$};
	\end{pgfonlayer}
\end{tikzpicture}}
\endpgfgraphicnamed}
\newcommand{\gdzo}{%
\beginpgfgraphicnamed{scalars//RZ01alpha}
\begin{tikzpicture}
	\begin{pgfonlayer}{nodelayer}
		\node [style=gn] (0) at (0.5, 0.1) {\footnotesize$\alpha$};
		\node [style=none] (1) at (0.5, -0.2) {};
	\end{pgfonlayer}
	\begin{pgfonlayer}{edgelayer}
		\draw (1.center) to (0.center);
	\end{pgfonlayer}

\end{tikzpicture}}
\endpgfgraphicnamed}
\newcommand{\gdzopi}{%
\beginpgfgraphicnamed{scalars//RZ01pi}
\begin{tikzpicture}
	\begin{pgfonlayer}{nodelayer}
		\node [style=srn] (0) at (0.5, 0.1) {\footnotesize$\pi$};
		\node [style=none] (1) at (0.5, -0.2) {};
	\end{pgfonlayer}
	\begin{pgfonlayer}{edgelayer}
		\draw (1.center) to (0.center);
	\end{pgfonlayer}

\end{tikzpicture}}
\endpgfgraphicnamed}
\newcommand{\gdzoz}{%
\beginpgfgraphicnamed{scalars//RZ01z}
\begin{tikzpicture}
	\begin{pgfonlayer}{nodelayer}
		\node [style=srn] (0) at (0.5, 0.1) {};
		\node [style=none] (1) at (0.5, -0.2) {};
	\end{pgfonlayer}
	\begin{pgfonlayer}{edgelayer}
		\draw (1.center) to (0.center);
	\end{pgfonlayer}

\end{tikzpicture}}
\endpgfgraphicnamed}
\newcommand{\gdoo}{%
\beginpgfgraphicnamed{scalars//RZ11alpha}
\begin{tikzpicture}
	\begin{pgfonlayer}{nodelayer}
		\node [style=gn] (0) at (0.5, 0) {\footnotesize$\alpha$};
		\node [style=none] (1) at (0.5, -0.3) {};
		\node [style=none] (2) at (0.5, 0.3) {};
	\end{pgfonlayer}
	\begin{pgfonlayer}{edgelayer}
		\draw (1.center) to (0.center);
				\draw (2.center) to (0.center);
	\end{pgfonlayer}

\end{tikzpicture}}
\endpgfgraphicnamed}
\newcommand{\gdto}{%
\beginpgfgraphicnamed{scalars//RZ21alpha}
\InputIfFileExists{scalars//RZ21alpha.tikz}{}{\input{./figures/scalars//RZ21alpha.tikz}}
\endpgfgraphicnamed}
\newcommand{\gpi}{%
\beginpgfgraphicnamed{scalars//RZ00pi}
\begin{tikzpicture}
	\begin{pgfonlayer}{nodelayer}
		\node [style=gn] (0) at (0.5, 0) {\footnotesize$\pi$};
	\end{pgfonlayer}
\end{tikzpicture}}
\endpgfgraphicnamed}
\newcommand{\rpi}{%
\beginpgfgraphicnamed{scalars//RZ00pir}
\begin{tikzpicture}
	\begin{pgfonlayer}{nodelayer}
		\node [style=rn] (0) at (0.5, 0) {\footnotesize$\pi$};
	\end{pgfonlayer}
\end{tikzpicture}}
\endpgfgraphicnamed}

\newcommand{\gdot}{%
\beginpgfgraphicnamed{scalars//RZ00zero}
\begin{tikzpicture}
	\begin{pgfonlayer}{nodelayer}
		\node [style=sgn] (0) at (0.5, 0) {};
	\end{pgfonlayer}
\end{tikzpicture}}
\endpgfgraphicnamed}

\newcommand{\rdzz}{%
\beginpgfgraphicnamed{scalars//RX00alpha}
\begin{tikzpicture}
	\begin{pgfonlayer}{nodelayer}
		\node [style=rn] (0) at (0.5, 0) {\footnotesize$\alpha$};
	\end{pgfonlayer}
\end{tikzpicture}}
\endpgfgraphicnamed}
\newcommand{\rdzo}{%
\beginpgfgraphicnamed{scalars//RX01alpha}
\begin{tikzpicture}
	\begin{pgfonlayer}{nodelayer}
		\node [style=rn] (0) at (0.5, 0.1) {\footnotesize$\alpha$};
		\node [style=none] (1) at (0.5, -0.2) {};
	\end{pgfonlayer}
	\begin{pgfonlayer}{edgelayer}
		\draw (1.center) to (0.center);
	\end{pgfonlayer}

\end{tikzpicture}}
\endpgfgraphicnamed}
\newcommand{\rdoo}{%
\beginpgfgraphicnamed{scalars//RX11alpha}
\begin{tikzpicture}
	\begin{pgfonlayer}{nodelayer}
		\node [style=rn] (0) at (0.5, 0) {\footnotesize$\alpha$};
		\node [style=none] (1) at (0.5, -0.3) {};
		\node [style=none] (2) at (0.5, 0.3) {};
	\end{pgfonlayer}
	\begin{pgfonlayer}{edgelayer}
		\draw (1.center) to (0.center);
				\draw (2.center) to (0.center);
	\end{pgfonlayer}

\end{tikzpicture}}
\endpgfgraphicnamed}

\newcommand{\den}[1]{\denoteb{#1}^\sharp}

\newcommand{\dens}[1]{\denote{#1}^\sharp}

\title{Supplementarity is Necessary for Quantum~Diagram~Reasoning}

\date{}
\author{Simon Perdrix\\CNRS, LORIA, Inria project team Carte, 
  Nancy, France\\\href{mailto:simon.perdrix@loria.fr}{simon.perdrix@loria.fr} \and Quanlong Wang\\ LORIA, Universit\'e de Lorraine, Nancy, France\\\href{mailto:quanlong.wang@loria.fr}{quanlong.wang@loria.fr}}
\begin{document}

\maketitle

\begin{abstract}
The ZX-calculus is a powerful diagrammatic language for quantum mechanics and quantum information processing. We prove that its $\frac \pi 4$-fragment is not complete, in other words the ZX-calculus is not complete for the so called  ``Clifford+T quantum mechanics''. The completeness of this fragment was one of the main open problems in \emph{categorical quantum mechanics}, a programme initiated by Abramsky and Coecke. The ZX-calculus was known to be incomplete for  quantum mechanics. On the other hand, its $\frac \pi 2$-fragment is known to be complete, i.e.~the ZX-calculus is complete for the so called ``stabilizer quantum mechanics''. Deciding whether its $\frac\pi 4$-fragment is complete is a crucial step in the development of the ZX-calculus since this fragment is approximately universal for quantum mechanics, contrary to  the $\frac \pi 2$-fragment. 

To establish our incompleteness result, we consider a fairly simple property of quantum states called supplementarity. 
We show that supplementarity can be derived in the ZX-calculus if and only if the angles involved in this equation are multiples of $\pi/2$. In particular, the impossibility to derive supplementarity for $\pi/4$ implies the incompleteness of the ZX-calculus for Clifford+T quantum mechanics. As a consequence, we propose to add the supplementarity to the set of rules of the ZX-calculus.

We also show that if a ZX-diagram involves antiphase twins, they can be merged when the ZX-calculus is augmented with the supplementarity rule. Merging antiphase twins  makes  diagrammatic reasoning much easier and provides a purely graphical meaning to the  supplementarity rule. 

 \end{abstract}


\section{Introduction}
The ZX-calculus has been  introduced by Coecke and Duncan  \cite{CoeckeDuncan} as a graphical language for pure state qubit quantum mechanics 
 where each diagram  can be interpreted as a linear map or a matrix in a typical way (so-called {\em standard interpretation}).
 Intuitively, a ZX-diagram is  made  of three kinds of vertices: \drd, \dgd, and \sdH, where each green or red vertex is parameterised by an angle.

Unlike the quantum circuit notation which has no transformation rules, the ZX-calculus combines the advantages of being intuitive with a built-in system of rewrite rules. These rewrite rules make the ZX-calculus into a formal system with nontrivial equalities between diagrams. As shown in \cite{CoeckeDuncan}, the ZX-calculus can be used to express any operation in pure state qubit quantum mechanics, i.e. it is {\em universal}. Furthermore, any equality derived in the ZX-calculus can also be derived in the standard  matrix mechanics, i.e. it is {\em sound}.

The converse of soundness  is {\em completeness}. Informally put, the ZX-calculus would be complete if any equality that can be derived using matrices can also be derived graphically. It has been shown in \cite{Vladimir} that the ZX-calculus is incomplete for the overall pure state qubit quantum mechanics, and there is no way on how to complete it by now. However, some fragments of the ZX-calculus are known to be complete. The $\frac \pi 2$-fragment, which corresponds to diagrams involving  angles  multiple of $\pi/2$, is complete \cite{Miriam1}. This fragment corresponds to the so called stabilizer quantum mechanics \cite{Gottesman}. The $\pi$-fragment is also complete \cite{perdrixduncan} and corresponds to real stabilizer quantum mechanics. 
 Meanwhile, the stabilizer completeness proof in \cite{Miriam1} carries over to a ZX-like
graphical calculus for Spekkens'  toy theory \cite{Miriam5}.

While it is an important and active area of research, stabilizer quantum mechanics is only a small part of all quantum mechanics. In particular stabilizer quantum mechanics is not universal, even approximately. This fragment is even efficiently simulatable  on a classical computer.
On the contrary, the $\frac \pi4$-fragment, which corresponds to the so-called ``Clifford+T quantum mechanics'' is approximately universal \cite{Boykin}: any unitary transformation can be approximated with an arbitrary precision by a diagram involving angles multiple of $\pi/4$ only. The completeness of the $\frac\pi4$-fragment is a crucial property and has even been stated as one of the major open questions in the categorical approach to quantum mechanics \cite{Miriam1,Miriam2,website}. A partial result has been proved in \cite{Miriam2}: the fragment composed of path diagrams involving angles multiple of $\pi/4$ is complete.

Our main contribution is to prove that the $\frac \pi 4$-fragment of the ZX-calculus is incomplete.  In other words, we prove that the ZX-calculus is not complete for the "Clifford+T quantum mechanics". To this end, we consider a simple equation called supplementarity. This equation is inspired by a work by Coecke and Edwards \cite{bobbill} on the structures of quantum entanglement. 
We show that supplementarity can be derived in the ZX-calculus if and only if the angles involved in this equation are multiples of $\pi/2$. In particular, the impossibility to derive this equation for $\pi/4$ implies the incompleteness of the ZX-calculus for ``Clifford+T quantum mechanics''.

We also show that in the ZX-calculus augmented with the supplementarity rule, antiphase twins can be merged. A pair of antiphase twins is a pair of vertices which have:  the same colour; the same neighbourhood; and antiphase angles (the difference between their angles is $\pi$). Merging antiphase twins makes diagrammatic reasoning much easier and provides a purely graphical meaning to the supplementarity rule.

Notice that various slightly different notions of soundness/completeness have been used so far in the context of the ZX-calculus, depending on whether the rules of the language should strictly preserve the standard interpretation (as used in this paper), or  up to a global phase, or even up to a (non-zero) scalar. Our result of incompleteness applies to any of these variants. However, we believe that the recent attempts to treat carefully the scalars and in particular the zero scalar are valuable, that is why we consider in this paper the strict notion of soundness and completeness. It should also be noticed that the notion of completeness used in the context of the ZX-calculus is different from a related one used in \cite{Selinger} to prove that finite dimensional Hilbert spaces are complete for dagger compact closed categories. The difference lies in that the concept of completeness used in the present paper is only concerned with the standard interpretation in finite dimensional Hilbert spaces, whereas, roughly speaking, in  \cite{Selinger} it is considered for every possible interpretation (of object variables as spaces and morphism variables as linear maps).

This paper is structured as follows: the ZX-calculus (diagrams, standard interpretation, and rules) is presented in section \ref{sec:zx}. Section \ref{sec:supplementarity} is dedicated to the supplementarity equation and its graphical  interpretation in terms of antiphase twins. In section \ref{sec:incomp} we show that supplementarity involving angles which are not multiples of $\pi/2$ cannot be derived in the ZX-calculus which implies the incompleteness of the $\frac \pi 4$-fragment.

 \section{ZX-calculus}\label{sec:zx}

\subsection{Diagrams and standard interpretation}

A ZX-diagram $D:k\to l$ with $k$ inputs and $l$ outputs is generated by:

\begin{center}
\begin{tabular}{|r@{~}r@{~}c@{~}lc|r@{~}r@{~}c@{~}lc|}
\hline
&$R_Z^{(n,m)}(\alpha)$&$:$&$n\to m$ & %
\beginpgfgraphicnamed{scalars//spideralpha}
\InputIfFileExists{scalars//spideralpha.tikz}{}{\input{./figures/scalars//spideralpha.tikz}}
\endpgfgraphicnamed & &$R_X^{(n,m)}(\alpha)$&$:$&$ n\to m$& %
\beginpgfgraphicnamed{scalars//spiderredalpha}
\InputIfFileExists{scalars//spiderredalpha.tikz}{}{\input{./figures/scalars//spiderredalpha.tikz}}
\endpgfgraphicnamed \\
\hline
& $H$&$:$&$1\to 1$ &%
\beginpgfgraphicnamed{scalars//Had4}
\InputIfFileExists{scalars//Had4.tikz}{}{\input{./figures/scalars//Had4.tikz}}
\endpgfgraphicnamed  &
 & $e $&$:$&$0 \to 0$ &%
\beginpgfgraphicnamed{scalars//emptysquare-small}
\InputIfFileExists{scalars//emptysquare-small.tikz}{}{\input{./figures/scalars//emptysquare-small.tikz}}
\endpgfgraphicnamed  \\\hline
 &$\mathbb I$&$:$&$1\to 1$&%
\beginpgfgraphicnamed{scalars//Id}
\begin{tikzpicture}
	\begin{pgfonlayer}{nodelayer}
		\node [style=none] (1) at (0.5, 0.3) {};
		\node [style=none] (2) at (0.5, -0.3) {};
		\node [style=none] (3) at (0.5, -0.5) {};
		\node [style=none] (4) at (0.5, 0.5) {};
	\end{pgfonlayer}
	\begin{pgfonlayer}{edgelayer}
		\draw (1.center) to (2.center);
	\end{pgfonlayer}
\end{tikzpicture}}
\endpgfgraphicnamed &
  &$\sigma$&$:$&$ 2\to 2$& %
\beginpgfgraphicnamed{scalars//swap}
\InputIfFileExists{scalars//swap.tikz}{}{\input{./figures/scalars//swap.tikz}}
\endpgfgraphicnamed \\\hline
  &$\epsilon$&$:$&$2\to 0$& %
\beginpgfgraphicnamed{scalars//cup}
\InputIfFileExists{scalars//cup.tikz}{}{\input{./figures/scalars//cup.tikz}}
\endpgfgraphicnamed&
  &$\eta$&$:$&$ 0\to 2$&  %
\beginpgfgraphicnamed{scalars//cap}
\InputIfFileExists{scalars//cap.tikz}{}{\input{./figures/scalars//cap.tikz}}
\endpgfgraphicnamed \\\hline
\end{tabular}\\where $m,n\in \mathbb N$ and $\alpha \in [0,2\pi)$
\end{center}

\begin{itemize}

\item Spacial composition: for any $D_1:a\to b$ and $D_2: c\to d$, $D_1\otimes D_2 : a+c\to b+d$ consists in placing $D_1$ and $D_2$ side-by-side, $D_2$ on the right of $D_1$.
\item Sequential composition: for any $D_1:a\to b$ and $D_2: b\to c$, $D_2\circ D_1 : a\to c$ consists in placing $D_1$ on the top of $D_2$, connecting the outputs of $D_1$ to the inputs of $D_2$.

\end{itemize}

When equal to $0$ modulo $2\pi$ the angles of the green and red dots are omitted:
$$%
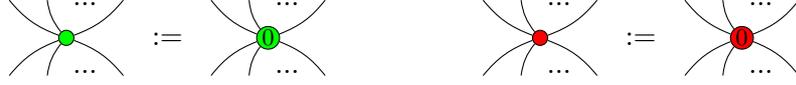
\beginpgfgraphicnamed{scalars//spiderg}
\InputIfFileExists{scalars//spiderg.tikz}{}{\input{./figures/scalars//spiderg.tikz}}
\endpgfgraphicnamed := %
\beginpgfgraphicnamed{scalars//spidergz}
\InputIfFileExists{scalars//spidergz.tikz}{}{\input{./figures/scalars//spidergz.tikz}}
\endpgfgraphicnamed\qquad\qquad%
\beginpgfgraphicnamed{scalars//spiderr}
\InputIfFileExists{scalars//spiderr.tikz}{}{\input{./figures/scalars//spiderr.tikz}}
\endpgfgraphicnamed := %
\beginpgfgraphicnamed{scalars//spiderrz}
\InputIfFileExists{scalars//spiderrz.tikz}{}{\input{./figures/scalars//spiderrz.tikz}}
\endpgfgraphicnamed\qquad$$

The standard interpretation of the ZX-diagrams associates with any diagram $D:n\to m$ a linear map $\denoteb{D}:\mathbb C^{2^n}\to \mathbb C^{2^m}$ inductively defined as follows:

\centerline{$
\def\arraystretch{0.5}
\denoteb{D_1\otimes D_2} := \denoteb{D_1}\otimes \denoteb{D_2}\qquad  \denoteb{D_2\circ D_1} := \denoteb{D_2}\times\denoteb{D_1}\qquad\denoteb{~\dempty~} := 1 \qquad \denoteb{~\did~}:= \left(\begin{array}{@{}c@{}c@{}}1&0\\0&1\end{array}\right)$}

\centerline{$
\def\arraystretch{0.5}
\denoteb{~\dHf~} := \frac{1}{\sqrt 2}\left(\begin{array}{@{}c@{}r@{}}1&1\\1&{~\text{-}}1\end{array}\right)\qquad\denoteb{~\dsigmaf~}:= \left(\begin{array}{@{}c@{}r@{}@{}c@{}c@{}}1&0&0&0\\0&0&1&0\\0&1&0&0\\0&0&0&1\\\end{array}\right)\qquad \denote{\drcupf}:= \left(\begin{array}{@{}c@{}c@{}@{}c@{}c@{}}1&0&0&1\end{array}\right) \qquad \denoteb{\drcapf}:= \left(\begin{array}{@{}c@{}}1\\0\\0\\1\end{array}\right) $}

\noindent $\denote{R_Z^{(0,0)}(\alpha)}:= 1{+}e^{i\alpha}$, and when $a{+}b>0$, $\denote{R_Z^{(a,b)}(\alpha)}$ is a matrix with $2^a$ columns and $2^b$ rows such that all entries are $0$ except the top left one which is $1$ and the bottom right one which is $e^{i\alpha}$, e.g.:
$$
\def\arraystretch{0.5}
\denoteb{~\gdzz~} = 1+e^{i\alpha} \qquad \denoteb{~\gdzo~} = \left(\begin{array}{@{}c@{}}1\\e^{i\alpha}\end{array}\right) \qquad \denoteb{~\gdoo~} = \left(\begin{array}{@{}c@{}c@{}}1&~0~\\0&~e^{i\alpha}\end{array}\right) \qquad  \denoteb{\gdto} = \left(\begin{array}{@{}c@{}c@{}c@{}c@{}}1~&~0~&~0~&~0~\\0~&~0~&~0~&~e^{i\alpha}\end{array}\right) $$
\noindent For any $a,b\ge  0$, $\denote{R_X^{a,b}(\alpha)}:= \denote{H}^{\otimes b}\times \denote{R_Z^{a,b}(\alpha)} \times \denote{H}^{\otimes a}$, where $M^{\otimes 0} =1$ and for any $k>0$, $M^{\otimes k}=M\otimes M^{\otimes k-1}$. E.g.,
$$
\def\arraystretch{0.5}
\denoteb{~\rdzz~} = 1+e^{i\alpha} \qquad \denoteb{~\rdzo~}= \sqrt2 e^{i\frac \alpha 2}\!\left(\begin{array}{@{}r@{}}\cos(\nicefrac{\alpha}2)\\\text{-}i\sin(\nicefrac{\alpha}2)\end{array}\right) \qquad \denoteb{~\rdoo~} = e^{i\frac \alpha 2}\! \left(\begin{array}{@{}r@{}r@{}}\cos(\nicefrac{\alpha}2)&\text{~~-}i\sin(\nicefrac{\alpha}2)\\\text{-}i\sin(\nicefrac{\alpha}2)&\cos(\nicefrac{\alpha}2)\end{array}\right) 
$$
ZX-diagrams are universal in the sense that for any $m,n\ge 0$ and any linear map $U:\mathbb C^{2^n}\to \mathbb C^{2^m}$, there exists a diagram $D:n\to m$ such that $\denote{D} = U$ \cite{CoeckeDuncan}. In particular, any unitary quantum evolution on a finite number of qubits can be represented by a ZX-diagram. Notice that universality implies to work with a uncountable set of angles. As a consequence, the approximate version of universality, i.e. the ability to approximate with arbitrary accuracy any linear map, is generally preferred in quantum information processing. 
The $\frac \pi 4$-fragment of language, which consists of all diagrams whose angles are multiples of $\pi/4$, is approximately universal, whereas the $\frac \pi2$-fragment is not. 

\subsection{Calculus}

The representation of a matrix in this graphical language is not unique. We present in this section the rules of the ZX calculus. 
These rules are sound in the sense that if two diagrams $D_1$ and $D_2$ are equal according to the rules of the ZX calculus, denoted $ZX\vdash D_1 = D_2$, then $\denote{D_1} = \denote{D_2}$. The rules of the language are given in Figure \ref{figure1}, and detailed bellow. The colour-swapped version and upside-down version of each rule given in Figure \ref{figure1} also apply.  

 \begin{figure}[!h]
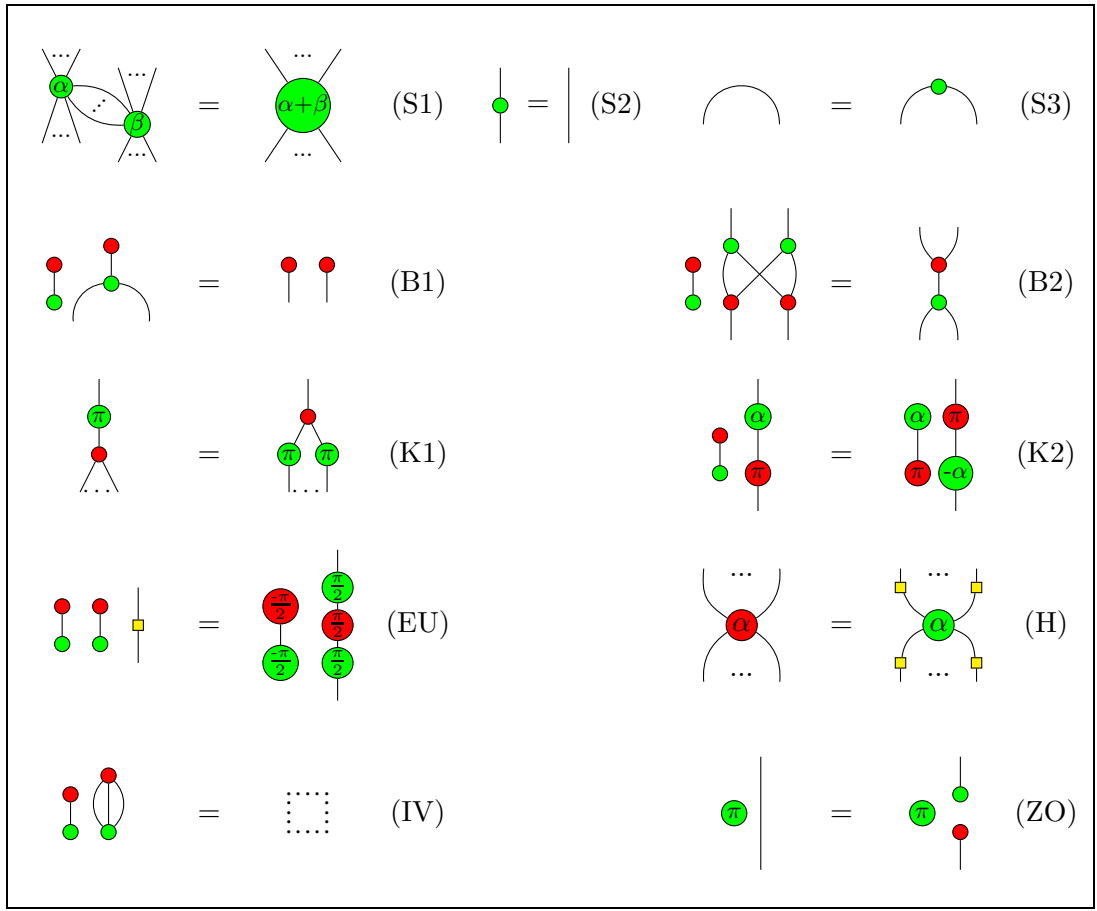

\begin{center}
 \centering
  \begin{tabular}{|ccccccccc|}
   \hline
   &&&&&&&& \\
\beginpgfgraphicnamed{scalars//spiderlt}
\InputIfFileExists{scalars//spiderlt.tikz}{}{\input{./figures/scalars//spiderlt.tikz}}
\endpgfgraphicnamed&=&%
\beginpgfgraphicnamed{scalars//spiderrt}
\InputIfFileExists{scalars//spiderrt.tikz}{}{\input{./figures/scalars//spiderrt.tikz}}
\endpgfgraphicnamed &(S1) &%
\beginpgfgraphicnamed{scalars//s2new}
\InputIfFileExists{scalars//s2new.tikz}{}{\input{./figures/scalars//s2new.tikz}}
\endpgfgraphicnamed (S2)& %
\beginpgfgraphicnamed{scalars//induced_compact_structure-2wirelt}
\InputIfFileExists{scalars//induced_compact_structure-2wirelt.tikz}{}{\input{./figures/scalars//induced_compact_structure-2wirelt.tikz}}
\endpgfgraphicnamed&=&%
\beginpgfgraphicnamed{scalars//induced_compact_structure-2wirert}
\InputIfFileExists{scalars//induced_compact_structure-2wirert.tikz}{}{\input{./figures/scalars//induced_compact_structure-2wirert.tikz}}
\endpgfgraphicnamed&(S3)\\
   &&&&&&&& \\
\beginpgfgraphicnamed{scalars//b1slt}
\InputIfFileExists{scalars//b1slt.tikz}{}{\input{./figures/scalars//b1slt.tikz}}
\endpgfgraphicnamed&=&%
\beginpgfgraphicnamed{scalars//b1srt}
\InputIfFileExists{scalars//b1srt.tikz}{}{\input{./figures/scalars//b1srt.tikz}}
\endpgfgraphicnamed&(B1) && %
\beginpgfgraphicnamed{scalars//b2slt}
\InputIfFileExists{scalars//b2slt.tikz}{}{\input{./figures/scalars//b2slt.tikz}}
\endpgfgraphicnamed&=&%
\beginpgfgraphicnamed{scalars//b2srt}
\InputIfFileExists{scalars//b2srt.tikz}{}{\input{./figures/scalars//b2srt.tikz}}
\endpgfgraphicnamed&(B2)\\
   &&&&&&&& \\
\beginpgfgraphicnamed{scalars//k1lt}
\InputIfFileExists{scalars//k1lt.tikz}{}{\input{./figures/scalars//k1lt.tikz}}
\endpgfgraphicnamed&=& %
\beginpgfgraphicnamed{scalars//k1rt}
\InputIfFileExists{scalars//k1rt.tikz}{}{\input{./figures/scalars//k1rt.tikz}}
\endpgfgraphicnamed&(K1) && %
\beginpgfgraphicnamed{scalars//k2slt}
\InputIfFileExists{scalars//k2slt.tikz}{}{\input{./figures/scalars//k2slt.tikz}}
\endpgfgraphicnamed&=&%
\beginpgfgraphicnamed{scalars//k2srt}
\InputIfFileExists{scalars//k2srt.tikz}{}{\input{./figures/scalars//k2srt.tikz}}
\endpgfgraphicnamed&(K2)\\
   &&&&&&&& \\
\beginpgfgraphicnamed{scalars//HadaDecomSingleslt}
\InputIfFileExists{scalars//HadaDecomSingleslt.tikz}{}{\input{./figures/scalars//HadaDecomSingleslt.tikz}}
\endpgfgraphicnamed&=&  %
\beginpgfgraphicnamed{scalars//HadaDecomSinglesrt}
\InputIfFileExists{scalars//HadaDecomSinglesrt.tikz}{}{\input{./figures/scalars//HadaDecomSinglesrt.tikz}}
\endpgfgraphicnamed&(EU) && %
\beginpgfgraphicnamed{scalars//h2newlt}
\InputIfFileExists{scalars//h2newlt.tikz}{}{\input{./figures/scalars//h2newlt.tikz}}
\endpgfgraphicnamed&=&%
\beginpgfgraphicnamed{scalars//h2newrt}
\InputIfFileExists{scalars//h2newrt.tikz}{}{\input{./figures/scalars//h2newrt.tikz}}
\endpgfgraphicnamed&(H)\\
   &&&&&&&& \\
\beginpgfgraphicnamed{scalars//inverserulelt}
\InputIfFileExists{scalars//inverserulelt.tikz}{}{\input{./figures/scalars//inverserulelt.tikz}}
\endpgfgraphicnamed&=&\dempty&(IV) && %
\beginpgfgraphicnamed{scalars//zo1lt}
\InputIfFileExists{scalars//zo1lt.tikz}{}{\input{./figures/scalars//zo1lt.tikz}}
\endpgfgraphicnamed&=&%
\beginpgfgraphicnamed{scalars//zo1rt}
\InputIfFileExists{scalars//zo1rt.tikz}{}{\input{./figures/scalars//zo1rt.tikz}}
\endpgfgraphicnamed&(ZO)\\
   &&&&&&&& \\
   \hline
  \end{tabular}
\end{center}

  \caption{Rules of ZX-calculus. The colour-swapped and/or upside-down versions of each rule also applies. Horizontal dots ($\ldots$) mean `arbitrary number', whereas diagonal dots ($\iddots$) mean `at least one'.}\label{figure1}
\end{figure}

  \noindent{\bf Spider.} \def\arraystretch{0.5}
  According to the  (S1) rule any two directly connected green dots can be merged. Moreover, a dot with a single input, single output and angle $0$ can be removed according to the (S2) rule.  These rules have their origins in the axiomatisation of orthonormal bases by means of the dagger special Frobenius algebras (see \cite{cpv} for details). According to the standard interpretation $\denote .$,  the green dots are associated with the so-called standard basis  
  $\{{1\choose 0},{0\choose 1}\}$, whereas the red dots (which also satisfies the spider property since colour-swapped rules also apply)  are associated with the so-called diagonal basis $\{\frac1{\sqrt2}\!{1\choose 1}, \frac1{\sqrt2}\!{1\choose \text{-}1}\}$.

 \noindent{\bf Green-Red Interactions.}  Monochromatic diagrams are lax:  according to the (S1) rule  any (green- or red-) monochromatic  connected diagram is equivalent to a single dot with the appropriate number of legs and whose angle is the sum of the angles. Thus the interesting structure arises when the two colours interact. The bialgebra rule (B1) and the copy rule (B2), imply that the red and the green bases are complementary, which roughly speaking captures the notion of uncertainty principle and  of unbiasedness, a fundamental property in quantum information (see \cite{CoeckeDuncan} for details).

 \noindent{\bf Parallel wires and Hopf law.}  (B1) and (B2) rules imply the following Hopf law \cite{CoeckeDuncan,DD}:
 $%
\beginpgfgraphicnamed{scalars//hopfgeneral1}
\InputIfFileExists{scalars//hopfgeneral1.tikz}{}{\input{./figures/scalars//hopfgeneral1.tikz}}
\endpgfgraphicnamed=%
\beginpgfgraphicnamed{scalars//hopfgeneral2}
\begin{tikzpicture}[scale=0.6]
	\begin{pgfonlayer}{nodelayer}
		\node [style=none] (0) at (0, 0.75) {};
		\node [style=srn] (1) at (0, -0.25) {};
		\node [style=none] (2) at (0, -0.75) {};
		\node [style=sgn] (3) at (0, 0.25) {};
	\end{pgfonlayer}
	\begin{pgfonlayer}{edgelayer}
		\draw (0.center) to (3);
		\draw (1) to (2.center);
	\end{pgfonlayer}
\end{tikzpicture}}
\endpgfgraphicnamed$ 
  where $%
\beginpgfgraphicnamed{scalars//dualizer}
\InputIfFileExists{scalars//dualizer.tikz}{}{\input{./figures/scalars//dualizer.tikz}}
\endpgfgraphicnamed:=%
\beginpgfgraphicnamed{scalars//dualizerdef}
\InputIfFileExists{scalars//dualizerdef.tikz}{}{\input{./figures/scalars//dualizerdef.tikz}}
\endpgfgraphicnamed$ is the called the antipode.  The (S3) rule trivialises the antipode and simplifies the Hopf law:\\
   {$\phantom{a}\qquad\quad\quad\quad\qquad\qquad\quad\quad\quad\qquad\quad\quad\quad%
\beginpgfgraphicnamed{scalars//hopfgeneral}
\InputIfFileExists{scalars//hopfgeneral.tikz}{}{\input{./figures/scalars//hopfgeneral.tikz}}
\endpgfgraphicnamed~=~%
\beginpgfgraphicnamed{scalars//hopfgeneral2}
\begin{tikzpicture}[scale=0.6]
	\begin{pgfonlayer}{nodelayer}
		\node [style=none] (0) at (0, 0.75) {};
		\node [style=srn] (1) at (0, -0.25) {};
		\node [style=none] (2) at (0, -0.75) {};
		\node [style=sgn] (3) at (0, 0.25) {};
	\end{pgfonlayer}
	\begin{pgfonlayer}{edgelayer}
		\draw (0.center) to (3);
		\draw (1) to (2.center);
	\end{pgfonlayer}
\end{tikzpicture}}
\endpgfgraphicnamed$\hfill{\bf (Hopf Law)}}

Hopf law has then a simple graphical meaning: two parallel wires between dots of distinct colours can be removed (up to the scalar $\gdot$). Notice that any pair of complementary basis in arbitrary finite dimension satisfies the rules  (S1), (S2), (B1) and (B2). However the  (S3) rule implies that the dimension of the corresponding Hilbert space is a power of two. As a consequence the ZX-calculus is a language dedicated to \emph{qubit quantum mechanics}. 


 \noindent{\bf Classical point.} In the context of complementary basis, the rules  (K1) and (K2) imply that \gdzopi is a classical point. Intuitively, it means that\gdzopi together with\gdzoz are two elements of the red basis, so in dimension $2$ they form an orthogonal basis. 

 \noindent{\bf Colour change.} According to the (H) rule, $\sdH$ can be used to change the colour of a dot. The (EU) rule corresponds to the Euler decomposition of the Hadamard matrix into three elementary rotations. 

 \noindent{\bf Scalar and zero.} A scalar is a diagram with no input and no output. The standard interpretation of such a diagram is  a complex number. While for simplicity, scalars have been ignored in several versions of the ZX calculus \cite{CoeckeDuncan, Miriam1},  recently several rules have been introduce for scalars \cite{Miriam3} and then simplified in \cite{bpw},  leading to the two rules (IV) and (ZO) presented in Figure \ref{figure1}. 
 As the interpretation of the empty diagram is $1$, the (IV) rule implies that $\diroot$ is the inverse of $\droot$. The interpretation of $\gpi$ is $0$, as a consequence for any diagrams $D_1$ and $D_2$, $\denote{\gpi \otimes D_1} = \denote{\gpi\otimes D_2}$. This absorbing property is captured by the  (ZO) rule.



\noindent{\bf Context.} The rules of the language presented in Figure \ref{figure1} can be applied to any subdiagram. In other words, if $ZX\vdash D_1 = D_2$ then, for any $D$ (with the appropriate number of inputs/ouputs), $ZX\vdash D\otimes D_1 = D\otimes D_2$ ;   $ZX\vdash D_1\otimes D = D_2\otimes D$ ; $ZX\vdash D\circ D_1 = D\circ D_2$ ; and  $ZX\vdash D_1\circ D = D_2\circ  D$.

 \noindent{\bf Only topology matters.}
A ZX-diagram can be deformed without changing its interpretation. This property is known as ``only topology matters'' in \cite{CoeckeDuncan}. E.g.
$$%
\beginpgfgraphicnamed{scalars//capswap}
\InputIfFileExists{scalars//capswap.tikz}{}{\input{./figures/scalars//capswap.tikz}}
\endpgfgraphicnamed ~=\drcapf (A) \qquad %
\beginpgfgraphicnamed{scalars//yanking1}
\InputIfFileExists{scalars//yanking1.tikz}{}{\input{./figures/scalars//yanking1.tikz}}
\endpgfgraphicnamed~ =~ %
\beginpgfgraphicnamed{scalars//yangkingline}
\begin{tikzpicture}[scale=0.6]
	\begin{pgfonlayer}{nodelayer}
		\node [style=none] (0) at (0, 0.5) {};
		\node [style=none] (1) at (0, -0.5) {};
	\end{pgfonlayer}
	\begin{pgfonlayer}{edgelayer}
		\draw (0.center) to (1.center);
	\end{pgfonlayer}
\end{tikzpicture}}
\endpgfgraphicnamed ~(B)\qquad%
\beginpgfgraphicnamed{scalars//commute1}
\InputIfFileExists{scalars//commute1.tikz}{}{\input{./figures/scalars//commute1.tikz}}
\endpgfgraphicnamed =~%
\beginpgfgraphicnamed{scalars//commute2}
\InputIfFileExists{scalars//commute2.tikz}{}{\input{./figures/scalars//commute2.tikz}}
\endpgfgraphicnamed~(C)\qquad %
\beginpgfgraphicnamed{scalars//bendingnew}
\InputIfFileExists{scalars//bendingnew.tikz}{}{\input{./figures/scalars//bendingnew.tikz}}
\endpgfgraphicnamed ~=%
\beginpgfgraphicnamed{scalars//nonbending}
\InputIfFileExists{scalars//nonbending.tikz}{}{\input{./figures/scalars//nonbending.tikz}}
\endpgfgraphicnamed ~(D)$$
%

%
%
 ``Only topology matters''  is a consequence of the underlying dagger compact closed structure (e.g. Eq. A and B), 
 together with the ability to interchange any two legs (Eq. C) and to turn inputs into outputs (Eq D) and vice-versa. Equations C and D are non standard in dagger compact closed categories, and are consequences of the other rules of the ZX-calculus \cite{bpw}. 

 \subsection{Soundness and Completeness}

 \noindent {\bf (In-)Completeness.}
 All the rules of the ZX calculus are sound with respect to the standard interpretation, i.e. if $ZX\vdash D_1 = D_2$ then $\denote{D_1} = \denote{D_2}$. The converse of soundness is completeness: the language would be complete if $\denote {D_1}  = \denote{D_2}$ implies $ZX\vdash D_1= D_2$. The completeness would imply that one can forget matrices and do graphical reasoning only. Completeness would also imply that  all the fundamental properties of qubit quantum mechanics are graphically captured by the rules of the ZX-calculus. This desirable property is one of the main open questions in categorical quantum mechanics. In the following , we review the known  results about the completeness of the ZX-calculus, which are essentially depending on the considered fragment (restriction on the angles) of the language.

 The very first result of incompleteness was about the original ZX-calculus in which the Euler decomposition\footnote{By Euler decomposition we mean the existence, for any $1$-qubit unitary $U$, of 4 angles $\alpha,\beta,\gamma,\delta$ s.t. $U = e^{i\alpha}R_x(\beta)R_z(\gamma)R_x(\delta)$ where $R_x(.)$ and $R_z(.)$ are elementary rotations about orthogonal axis.} of $H$, the (EU) rule in Figure \ref{figure1} was not derivable. This equation is now part of the language. Backens \cite{Miriam1} proved that the $\frac \pi 2$ fragment is complete.  Schr\"oder and Zamdzhiev proved that the language is not complete in general.  Their argument is also based on some  Euler decomposition, but contrary to the previous case this decomposition involves non rational multiples of $\pi$. The most natural way -- and actually the only known way -- to bypass this incompleteness result is to consider a fragment of the language.
  Indeed, irrational multiples of $\pi$ are not necessary for approximate universality. As the $\frac \pi2$-fragment is not approximately universal, the most interesting candidate for completeness is the $\frac \pi4$-fragment which is approximately universal. The completeness for the $\frac \pi 4$-fragment has been conjectured in \cite{Miriam2} and actually proved in the  single qubit case, i.e. for path diagrams. The use of path diagrams (diagrams with all dots of degree two) is rather restrictive, but the completeness for this class of diagrams is not trivial and is sufficient to show that any argument based on some Euler decomposition cannot be applied in the $\pi/4$ case. 
However, we disprove the conjecture:  the $\frac \pi 4$-fragment of the ZX-calculus is not complete (corollary \ref{cor:incomp}), using a novel approach not based on Euler decompositions.

 \noindent{\bf Scalars and completeness.}
 In several versions of the ZX-calculus scalars are ignored, leading to a slightly different notion of soundness and completeness involving proportionality. Roughly speaking, ignoring the scalars consists in an additional rule which allows  one to freely  add or remove diagram with no input/output. A particular attention has to be paid to `zero' diagrams, i.e. diagrams whose interpretations are zero, like $\gpi$. When scalars are ignored, the notion of soundness is modified as follows: if $D_1=D_2$ then  $\denote{D_1}$ and $\denote{D_2}$ are proportional. The definition of completeness is modified likewise.  Notice that in \cite{Vladimir} yet another notion of soundness is considered where scalars are not ignored in general but global phases are, i.e. if $D_1=D_2$ then $\exists \theta, \denote{D_1} = e^{i\theta}\denote{D_2}$. Our main result of incompleteness (Theorem \ref{thm:incomp}) applies for any of these variants of soundness/completeness. However, we believe that the recent attempts to treat carefully the scalars and in particular the zero scalar are valuable, that is why we consider in this paper the strict notion of soundness and completeness.

 \section{Supplementarity and antiphase twins}\label{sec:supplementarity}

\newcommand{\splus}{\raisebox{0.03cm}{\tiny\textsf{+}}}
\newcommand{\sminus}{\raisebox{0.03cm}{\tiny\textsf{-}}}

In \cite{bobbill}, Coecke and Edwards introduced the notion of \emph{supplementarity} by pointing out that when $\alpha\neq 0\bmod \pi$ the standard interpretation of the following diagram is proportional to the projector $10 \choose 00$ if $\alpha-\beta = \pi$ and to the projector $00\choose 01$ if $\alpha+\beta = \pi$.
$${%
\beginpgfgraphicnamed{scalars//prop0}
\InputIfFileExists{scalars//prop0.tikz}{}{\input{./figures/scalars//prop0.tikz}}
\endpgfgraphicnamed} $$
Putting back the scalars, one gets the following equations, which are true for any angle $\alpha$, even when $\alpha = 0$:
$$
\left\llbracket%
\beginpgfgraphicnamed{scalars//prop1sc}
\InputIfFileExists{scalars//prop1sc.tikz}{}{\input{./figures/scalars//prop1sc.tikz}}
\endpgfgraphicnamed\right\rrbracket=\left\llbracket%
\beginpgfgraphicnamed{scalars//prop2sc}
\InputIfFileExists{scalars//prop2sc.tikz}{}{\input{./figures/scalars//prop2sc.tikz}}
\endpgfgraphicnamed\right\rrbracket
~~~\text{and}~~~ \left\llbracket%
\beginpgfgraphicnamed{scalars//prop5sc}
\InputIfFileExists{scalars//prop5sc.tikz}{}{\input{./figures/scalars//prop5sc.tikz}}
\endpgfgraphicnamed\right\rrbracket=\left\llbracket%
\beginpgfgraphicnamed{scalars//prop4sc}
\InputIfFileExists{scalars//prop4sc.tikz}{}{\input{./figures/scalars//prop4sc.tikz}}
\endpgfgraphicnamed\right\rrbracket.
$$
Albeit Coecke and Edwards did not address explicitly the question of proving whether these equations can be derived in the ZX-calculus or not, these equations were known to be candidates for proving the incompleteness of the language\footnote{Personnal communications with Miriam Backens and Aleks Kissinger.}. We prove in section \ref{sec:incomp} that these equations can be derived in the ZX-calculus only when $\alpha = 0\bmod \frac \pi 2$.

Inspired by the property pointed out by Coecke and Edwards we introduce the following equation that we call \emph{supplementarity}:
\begin{equation}\label{eqn:supp}
\hspace{4.5cm}%
\beginpgfgraphicnamed{scalars//suppplenew}
\InputIfFileExists{scalars//suppplenew.tikz}{}{\input{./figures/scalars//suppplenew.tikz}}
\endpgfgraphicnamed
\end{equation}

Supplementarity is sound in the sense that both diagrams of (Eq. \ref{eqn:supp}) have the same standard interpretation $\frac 1{\sqrt 2}{1-e^{2i\alpha} \choose 0}$. It is provable in the ZX-calculus that supplementarity  (Eq. \ref{eqn:supp}) is equivalent to the equations pointed out by Coecke and Edwards:

\begin{lemma}\label{lem:sup} In the ZX calculus, for any $\alpha\in [0, 2\pi)$:
\emph{\begin{equation*}\label{equiv}
\beginpgfgraphicnamed{scalars//suppplenew}
\InputIfFileExists{scalars//suppplenew.tikz}{}{\input{./figures/scalars//suppplenew.tikz}}
\endpgfgraphicnamed~~
\qquad\Leftrightarrow \qquad
\beginpgfgraphicnamed{scalars//prop1sc}
\InputIfFileExists{scalars//prop1sc.tikz}{}{\input{./figures/scalars//prop1sc.tikz}}
\endpgfgraphicnamed~=~%
\beginpgfgraphicnamed{scalars//prop2sc}
\InputIfFileExists{scalars//prop2sc.tikz}{}{\input{./figures/scalars//prop2sc.tikz}}
\endpgfgraphicnamed
\qquad\Leftrightarrow\qquad %
\beginpgfgraphicnamed{scalars//prop5sc}
\InputIfFileExists{scalars//prop5sc.tikz}{}{\input{./figures/scalars//prop5sc.tikz}}
\endpgfgraphicnamed~=~%
\beginpgfgraphicnamed{scalars//prop4sc}
\InputIfFileExists{scalars//prop4sc.tikz}{}{\input{./figures/scalars//prop4sc.tikz}}
\endpgfgraphicnamed
\end{equation*}}

\end{lemma}

The proof of Lemma \ref{lem:sup} is given in appendix.

Graphically, the supplementarity equation can be interpreted  as merging two dots in a particular configuration: they are antiphase (i.e. same colour and the difference between the two angles is $\pi$); of degree one; and they have the same neighbour. While antiphase is a necessary condition, the other conditions can be relaxed to any ``twins'' as follows:

\begin{definition}[Antiphase Twins] Two dots $u$ and $v$ in a ZX-diagram are antiphase twins if:
\begin{itemize}
\item they have the same colour;
\item the difference between their angles is $\pi$;
\item they have the same neighbourhood: for any other vertex  (\drd, \dgd or \sdH) $w$, the number of wires connecting $u$ to $w$, and $v$ to $w$ are the same. 
\end{itemize}

\end{definition}


Notice that antiphase twins might be directly connected or not. Here two examples of antiphase twins and how they merge:
$$%
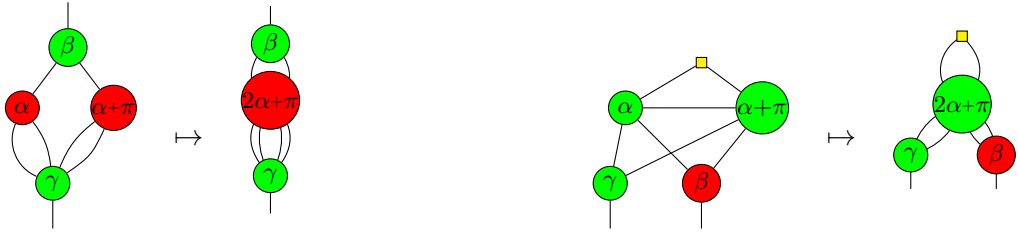
\beginpgfgraphicnamed{scalars//antiphasexb1}
\InputIfFileExists{scalars//antiphasexb1.tikz}{}{\input{./figures/scalars//antiphasexb1.tikz}}
\endpgfgraphicnamed  \quad\mapsto\quad   %
\beginpgfgraphicnamed{scalars//antiphasexb2}
\InputIfFileExists{scalars//antiphasexb2.tikz}{}{\input{./figures/scalars//antiphasexb2.tikz}}
\endpgfgraphicnamed   \qquad\qquad\qquad\qquad\qquad   %
\beginpgfgraphicnamed{scalars//antiphasexa1}
\InputIfFileExists{scalars//antiphasexa1.tikz}{}{\input{./figures/scalars//antiphasexa1.tikz}}
\endpgfgraphicnamed  \quad\mapsto\quad     %
\beginpgfgraphicnamed{scalars//antiphasexa2}
\InputIfFileExists{scalars//antiphasexa2.tikz}{}{\input{./figures/scalars//antiphasexa2.tikz}}
\endpgfgraphicnamed
$$

\begin{theorem}[Antiphase Twins and Supplementarity]
In ZX-calculus, antiphase twins can be merged if and only if $\forall \alpha, %
\beginpgfgraphicnamed{scalars//ssuppplenew}
\InputIfFileExists{scalars//ssuppplenew.tikz}{}{\input{./figures/scalars//ssuppplenew.tikz}}
\endpgfgraphicnamed$.
\end{theorem}

\begin{proof}
$[\Rightarrow]$ Supplementarity equation can be proved by merging antiphase twins.\\
$[\Leftarrow]$ Let $u$ and $v$ be antiphase twins. We assume w.l.o.g. that $u$ and $v$ are green dots. \\(a) If $u$ and $v$ are neighbours or if one of their neighbours is a green dot then $u$ and $v$ can be merged thanks to the (S1) rule, like in the following example. Notice that in this case supplementarity is not used.
$$%
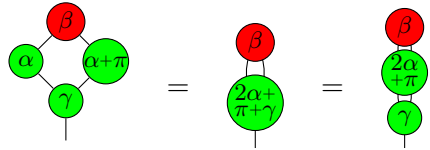
\beginpgfgraphicnamed{scalars//antiphasexc1}
\InputIfFileExists{scalars//antiphasexc1.tikz}{}{\input{./figures/scalars//antiphasexc1.tikz}}
\endpgfgraphicnamed  \quad =\quad  %
\beginpgfgraphicnamed{scalars//antiphasexc2}
\InputIfFileExists{scalars//antiphasexc2.tikz}{}{\input{./figures/scalars//antiphasexc2.tikz}}
\endpgfgraphicnamed \quad= \quad  %
\beginpgfgraphicnamed{scalars//antiphasexc3}
\InputIfFileExists{scalars//antiphasexc3.tikz}{}{\input{./figures/scalars//antiphasexc3.tikz}}
\endpgfgraphicnamed
$$
\noindent(b) If all neighbours of $u$ and $v$ are red dots, the sub-diagram induced by $u$, $v$ and their neighbours is a complete bipartite green red diagram which can be simplified using the following generalised bialgebra equation, proved (without the scalars) in \cite{perdrixduncan2}:

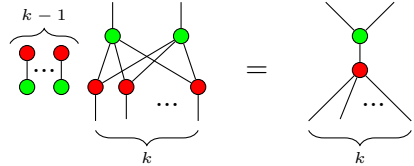
\centerline{%
\beginpgfgraphicnamed{scalars//gbia1}
\InputIfFileExists{scalars//gbia1.tikz}{}{\input{./figures/scalars//gbia1.tikz}}
\endpgfgraphicnamed\quad =\quad%
\beginpgfgraphicnamed{scalars//gbia3}
\InputIfFileExists{scalars//gbia3.tikz}{}{\input{./figures/scalars//gbia3.tikz}}
\endpgfgraphicnamed}
\noindent where $k>0$ is the number of outputs.
This equation on complete bipartite sub-diagrams can be used to transform antiphase twins into a configuration where supplementarity can be applied, and then back to a sub-diagram where the antiphase twins have been merged like in the following example:
$$%
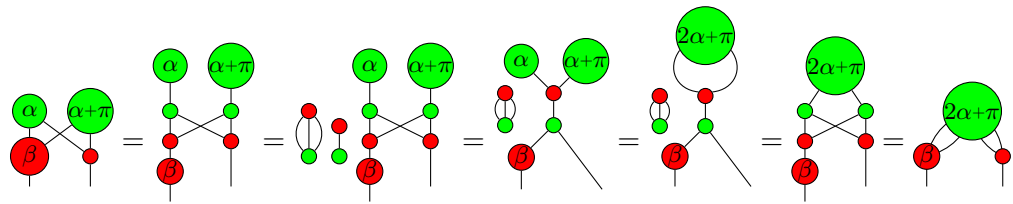
\beginpgfgraphicnamed{scalars//biasup1}
\InputIfFileExists{scalars//biasup1.tikz}{}{\input{./figures/scalars//biasup1.tikz}}
\endpgfgraphicnamed=%
\beginpgfgraphicnamed{scalars//biasup2}
\InputIfFileExists{scalars//biasup2.tikz}{}{\input{./figures/scalars//biasup2.tikz}}
\endpgfgraphicnamed=%
\beginpgfgraphicnamed{scalars//biasup3}
\InputIfFileExists{scalars//biasup3.tikz}{}{\input{./figures/scalars//biasup3.tikz}}
\endpgfgraphicnamed=%
\beginpgfgraphicnamed{scalars//biasup4}
\InputIfFileExists{scalars//biasup4.tikz}{}{\input{./figures/scalars//biasup4.tikz}}
\endpgfgraphicnamed=%
\beginpgfgraphicnamed{scalars//biasup5}
\InputIfFileExists{scalars//biasup5.tikz}{}{\input{./figures/scalars//biasup5.tikz}}
\endpgfgraphicnamed
=%
\beginpgfgraphicnamed{scalars//biasup6}
\InputIfFileExists{scalars//biasup6.tikz}{}{\input{./figures/scalars//biasup6.tikz}}
\endpgfgraphicnamed=%
\beginpgfgraphicnamed{scalars//biasup7}
\InputIfFileExists{scalars//biasup7.tikz}{}{\input{./figures/scalars//biasup7.tikz}}
\endpgfgraphicnamed$$

\noindent(c) If at least one of the neighbours of $u$ and $v$ is a \sdH, one can use the (EU) rule to decompose \sdH into green and red dots, then merge anti-phase twins as in case (b), and finally apply the  (EU) rule the other way around to reconstruct \sdH, like in the example below:
$$%
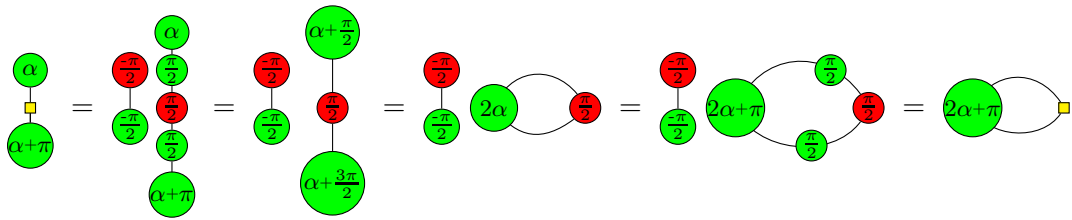
\beginpgfgraphicnamed{scalars/hadsup1}
\begin{tikzpicture}\footnotesize
	\begin{pgfonlayer}{nodelayer}
		\node [style=gn] (0) at (0, -0.5) {$\alpha{\splus}\pi$};
		\node [style={H box}] (1) at (0, -0) {};
		\node [style=gn] (2) at (0, 0.5) {$~\alpha~$};
	\end{pgfonlayer}
	\begin{pgfonlayer}{edgelayer}
		\draw (1) to (0);
		\draw (2) to (0);
	\end{pgfonlayer}
\end{tikzpicture}}
\endpgfgraphicnamed~=~%
\beginpgfgraphicnamed{scalars/eurulescalar}
\begin{tikzpicture}\footnotesize
	\begin{pgfonlayer}{nodelayer}
		\node [style=rn] (0) at (0, 0.5) {$\frac{\textnormal{-}\pi}{2}$};
		\node [style=gn] (1) at (0, -0.25) {$\frac{\textnormal{-}\pi}{2}$};
	\end{pgfonlayer}
	\begin{pgfonlayer}{edgelayer}
		\draw (0) to (1);
	\end{pgfonlayer}
\end{tikzpicture}}
\endpgfgraphicnamed%
\beginpgfgraphicnamed{scalars/hadsup2}
\InputIfFileExists{scalars/hadsup2.tikz}{}{\input{./figures/scalars/hadsup2.tikz}}
\endpgfgraphicnamed~=~%
\beginpgfgraphicnamed{scalars/eurulescalar}
\begin{tikzpicture}\footnotesize
	\begin{pgfonlayer}{nodelayer}
		\node [style=rn] (0) at (0, 0.5) {$\frac{\textnormal{-}\pi}{2}$};
		\node [style=gn] (1) at (0, -0.25) {$\frac{\textnormal{-}\pi}{2}$};
	\end{pgfonlayer}
	\begin{pgfonlayer}{edgelayer}
		\draw (0) to (1);
	\end{pgfonlayer}
\end{tikzpicture}}
\endpgfgraphicnamed~%
\beginpgfgraphicnamed{scalars/hadsup3}
\InputIfFileExists{scalars/hadsup3.tikz}{}{\input{./figures/scalars/hadsup3.tikz}}
\endpgfgraphicnamed~=~%
\beginpgfgraphicnamed{scalars/eurulescalar}
\begin{tikzpicture}\footnotesize
	\begin{pgfonlayer}{nodelayer}
		\node [style=rn] (0) at (0, 0.5) {$\frac{\textnormal{-}\pi}{2}$};
		\node [style=gn] (1) at (0, -0.25) {$\frac{\textnormal{-}\pi}{2}$};
	\end{pgfonlayer}
	\begin{pgfonlayer}{edgelayer}
		\draw (0) to (1);
	\end{pgfonlayer}
\end{tikzpicture}}
\endpgfgraphicnamed~%
\beginpgfgraphicnamed{scalars/hadsup4}
\InputIfFileExists{scalars/hadsup4.tikz}{}{\input{./figures/scalars/hadsup4.tikz}}
\endpgfgraphicnamed~=~%
\beginpgfgraphicnamed{scalars/eurulescalar}
\begin{tikzpicture}\footnotesize
	\begin{pgfonlayer}{nodelayer}
		\node [style=rn] (0) at (0, 0.5) {$\frac{\textnormal{-}\pi}{2}$};
		\node [style=gn] (1) at (0, -0.25) {$\frac{\textnormal{-}\pi}{2}$};
	\end{pgfonlayer}
	\begin{pgfonlayer}{edgelayer}
		\draw (0) to (1);
	\end{pgfonlayer}
\end{tikzpicture}}
\endpgfgraphicnamed~%
\beginpgfgraphicnamed{scalars/hadsup5}
\InputIfFileExists{scalars/hadsup5.tikz}{}{\input{./figures/scalars/hadsup5.tikz}}
\endpgfgraphicnamed~=~%
\beginpgfgraphicnamed{scalars/hadsup6}
\InputIfFileExists{scalars/hadsup6.tikz}{}{\input{./figures/scalars/hadsup6.tikz}}
\endpgfgraphicnamed$$
\end{proof}

\section{Supplementarity is necessary}\label{sec:incomp}

In this section, we prove the main result of the paper: supplementarity involving angles which are not multiples of $\frac \pi2$ cannot be derived using the rules of the ZX-calculus, and as a corollary the $\frac \pi4$-fragment of ZX-calculus is incomplete.

\begin{theorem}\label{thm:incomp}Supplementarity can be derived in the ZX-calculus only for multiples of $\pi/2$:
  \emph{$$\left(ZX\vdash %
\beginpgfgraphicnamed{scalars//suppplenew}
\InputIfFileExists{scalars//suppplenew.tikz}{}{\input{./figures/scalars//suppplenew.tikz}}
\endpgfgraphicnamed\right) \qquad \Leftrightarrow \qquad \alpha = 0\bmod \frac \pi2 $$}
\end{theorem}

\begin{corollary}\label{cor:incomp}
The $\frac \pi 4$-fragment of ZX-calculus is not complete. In other words, ZX-calculus is not complete for the so-called ``Clifford+T quantum mechanics''.
\end{corollary}

The rest of the section is dedicated to the proof of Theorem \ref{thm:incomp}. To do so,  we introduce an alternative interpretation $\den{.}$ for the diagrams, that we prove to be sound (Lemma \ref{lem:sound}) but  for which $\den{%
\beginpgfgraphicnamed{scalars//ssuppplenew-left}
\InputIfFileExists{scalars//ssuppplenew-left.tikz}{}{\input{./figures/scalars//ssuppplenew-left.tikz}}
\endpgfgraphicnamed}\neq \den{%
\beginpgfgraphicnamed{scalars//ssuppplenew-right}
\InputIfFileExists{scalars//ssuppplenew-right.tikz}{}{\input{./figures/scalars//ssuppplenew-right.tikz}}
\endpgfgraphicnamed}$ when $\alpha \neq 0\bmod \frac \pi 2$.

\begin{definition}For any diagram $D: n\to m$, let $\den{D}: 3n\to 3m$ be a diagram defined as follows:

\centerline{$
\def\arraystretch{0.5}
\denote{D_1\otimes D_2}^\sharp := \denote{D_1}^\sharp{\otimes} \denote{D_2}^\sharp\quad  ~\denote{D_2\circ D_1}^\sharp := \den{D_2}{\times}\den{D_1}\quad~\den{~\dempty~} := \dempty \quad~ \den{~\didf~}:=~\did~\did~\did$}

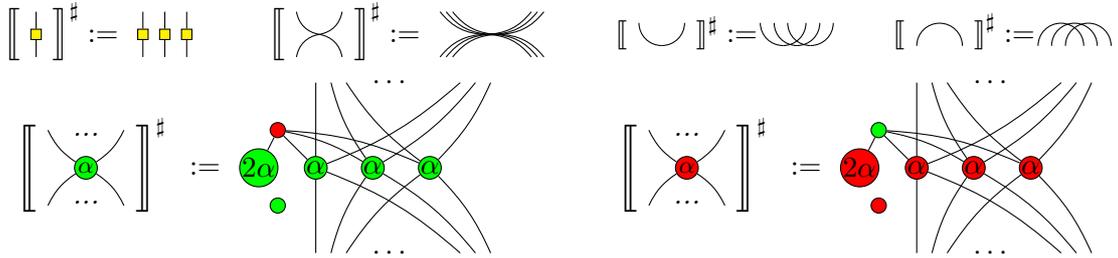
\centerline{$
\def\arraystretch{0.5}
\den{~\dHf~} := ~\dH~\dH~\dH~\qquad\den{~\dsigmaf~}:= ~\ldsigma\!\!\!\!\!\!\!\!\!\!\!\!\!\!\!\!\!\!\ldsigma\!\!\!\!\!\!\!\!\!\!\!\!\!\!\!\!\!\!\ldsigma~\qquad \dens{\drcupf}:=  ~\drcup\!\!\!\!\!\!\!\!\!\!\!\!\!\!\!\!\!\drcup\!\!\!\!\!\!\!\!\!\!\!\!\!\!\!\!\!\drcup~ \qquad \den{\drcapf}:= ~\drcap\!\!\!\!\!\!\!\!\!\!\!\!\!\!\!\!\!\drcap\!\!\!\!\!\!\!\!\!\!\!\!\!\!\!\!\!\drcap~ $}

\centerline{$\den{~%
\beginpgfgraphicnamed{scalars//spiderg-a}
\InputIfFileExists{scalars//spiderg-a.tikz}{}{\input{./figures/scalars//spiderg-a.tikz}}
\endpgfgraphicnamed~}  ~:=~ %
\beginpgfgraphicnamed{scalars//spiderinterpretationsc}
\InputIfFileExists{scalars//spiderinterpretationsc.tikz}{}{\input{./figures/scalars//spiderinterpretationsc.tikz}}
\endpgfgraphicnamed\qquad\qquad\den{~%
\beginpgfgraphicnamed{scalars//spiderr-a}
\InputIfFileExists{scalars//spiderr-a.tikz}{}{\input{./figures/scalars//spiderr-a.tikz}}
\endpgfgraphicnamed~}  ~:=~ %
\beginpgfgraphicnamed{scalars//spiderinterpretationsc-r}
\InputIfFileExists{scalars//spiderinterpretationsc-r.tikz}{}{\input{./figures/scalars//spiderinterpretationsc-r.tikz}}
\endpgfgraphicnamed$}

\end{definition}

Roughly speaking, $\den D$ consists of three copies of $D$ together with, for each dot of angle $\alpha$, a \emph{gadget} parameterized by the angle $2\alpha$ connecting the three copies of the dot. E.g.
$$\den{~%
\beginpgfgraphicnamed{scalars//alpha}
\begin{tikzpicture}[scale=0.5]\footnotesize
	\begin{pgfonlayer}{nodelayer}
		\node [style=none] (0) at (0, -0.75) {};
		\node [style=gn] (1) at (0, -0) {$\alpha$};
		\node [style=none] (2) at (0, 0.75) {};
	\end{pgfonlayer}
	\begin{pgfonlayer}{edgelayer}
		\draw (2.center) to (1);
		\draw (1) to (0.center);
	\end{pgfonlayer}
\end{tikzpicture}}
\endpgfgraphicnamed~}=%
\beginpgfgraphicnamed{scalars//alphainterpretationsc}
\InputIfFileExists{scalars//alphainterpretationsc.tikz}{}{\input{./figures/scalars//alphainterpretationsc.tikz}}
\endpgfgraphicnamed$$
Simple calculations show that the gadget disappears when $\alpha=0\bmod\pi$, e.g.:
$$\den{%
\beginpgfgraphicnamed{scalars//sgreendelta}
\InputIfFileExists{scalars//sgreendelta.tikz}{}{\input{./figures/scalars//sgreendelta.tikz}}
\endpgfgraphicnamed} = %
\beginpgfgraphicnamed{scalars//s3grdelta}
\InputIfFileExists{scalars//s3grdelta.tikz}{}{\input{./figures/scalars//s3grdelta.tikz}}
\endpgfgraphicnamed \qquad\qquad \den{~%
\beginpgfgraphicnamed{scalars//szdot}
\begin{tikzpicture}[scale=0.5]
	\begin{pgfonlayer}{nodelayer}
		\node [style=gn] (0) at (0.5, 0.25) {$\pi$};
		\node [style=none] (1) at (0.5, -0.5) {};
	\end{pgfonlayer}
	\begin{pgfonlayer}{edgelayer}
		\draw (0) to (1.center);
	\end{pgfonlayer}
\end{tikzpicture}}
\endpgfgraphicnamed~} = %
\beginpgfgraphicnamed{scalars//sinpzdot}
\InputIfFileExists{scalars//sinpzdot.tikz}{}{\input{./figures/scalars//sinpzdot.tikz}}
\endpgfgraphicnamed$$

\begin{lemma}[Soundness]\label{lem:sound}
$\denote{.}^\sharp$ is a sound interpretation: if $ZX \vdash D_1=D_2$ then $ZX\vdash \den{D_1} = \den{D_2}$.
\end{lemma}

\begin{proof}
Soundness is trivial for the $\pi$-fragment of the language  (i.e. when angles are multiples of $\pi$). Thus, it remains the four rules $(S1)$, $(K2)$, $(EU)$, and $(H)$ to complete the proof. We give the proof of $(K2)$ and a particular case of $(S1)$ to illustrate the proof, the other cases are given in appendix.

\noindent[(K2)]
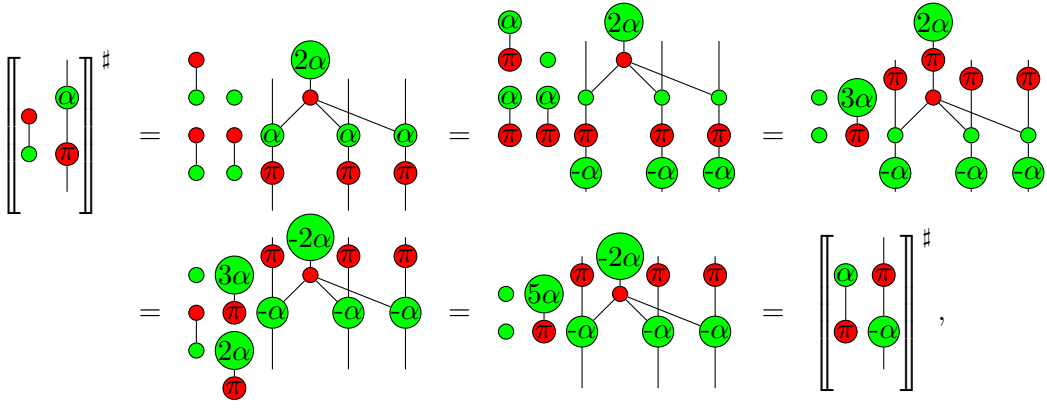
\begin{align*}\begin{array}{lllllll}
\left\llbracket%
\beginpgfgraphicnamed{scalars//alphapisc}
\InputIfFileExists{scalars//alphapisc.tikz}{}{\input{./figures/scalars//alphapisc.tikz}}
\endpgfgraphicnamed\right\rrbracket^{\sharp}&=&%
\beginpgfgraphicnamed{scalars//interprepialphapf1}
\InputIfFileExists{scalars//interprepialphapf1.tikz}{}{\input{./figures/scalars//interprepialphapf1.tikz}}
\endpgfgraphicnamed&=&%
\beginpgfgraphicnamed{scalars//interprepialphapf2}
\InputIfFileExists{scalars//interprepialphapf2.tikz}{}{\input{./figures/scalars//interprepialphapf2.tikz}}
\endpgfgraphicnamed
&=&%
\beginpgfgraphicnamed{scalars//interprepialphapf3}
\InputIfFileExists{scalars//interprepialphapf3.tikz}{}{\input{./figures/scalars//interprepialphapf3.tikz}}
\endpgfgraphicnamed\\
&=&%
\beginpgfgraphicnamed{scalars//interprepialphapf4}
\InputIfFileExists{scalars//interprepialphapf4.tikz}{}{\input{./figures/scalars//interprepialphapf4.tikz}}
\endpgfgraphicnamed&=&%
\beginpgfgraphicnamed{scalars//interprepialphapf5}
\InputIfFileExists{scalars//interprepialphapf5.tikz}{}{\input{./figures/scalars//interprepialphapf5.tikz}}
\endpgfgraphicnamed&=&\left\llbracket%
\beginpgfgraphicnamed{scalars//piminusalphasc}
\InputIfFileExists{scalars//piminusalphasc.tikz}{}{\input{./figures/scalars//piminusalphasc.tikz}}
\endpgfgraphicnamed\right\rrbracket^{\sharp},
\end{array}
\end{align*}
The first equality is nothing but the definition of $\den.$.
Second step is based on the (K2) rule. The third step consists in (i) grouping the 3 scalars depending on $\alpha$ into a single one, to do so rules (B1), (K1) and finally (S1) are combined; (ii) applying the (K1) rule on the non scalar part of the diagram. Fourth step consists in applying the (K2) rule on the gadget. The fifth step is combining the scalars depending on $\alpha$. Finally for the last step we use $\left\llbracket%
\beginpgfgraphicnamed{scalars//galpharpi}
\begin{tikzpicture}
	\begin{pgfonlayer}{nodelayer}
		\node [style=gn] (0) at (0, 0.25) {$\alpha$};
		\node [style=rn] (1) at (0, -0.25) {$\pi$};
	\end{pgfonlayer}
	\begin{pgfonlayer}{edgelayer}
		\draw (0) to (1);
	\end{pgfonlayer}
\end{tikzpicture}}
\endpgfgraphicnamed\right\rrbracket^{\sharp}=%
\beginpgfgraphicnamed{scalars//g5alpharpi}
\begin{tikzpicture}
	\begin{pgfonlayer}{nodelayer}
		\node [style=gn] (0) at (0, 0.25) {$5\alpha$};
		\node [style=sgn] (1) at (-0.5, -0) {};
		\node [style=rn] (2) at (0, -0.25) {$\pi$};
	\end{pgfonlayer}
	\begin{pgfonlayer}{edgelayer}
		\draw (0) to (2);
	\end{pgfonlayer}
\end{tikzpicture}}
\endpgfgraphicnamed$ .

\noindent [(S1)] In the following we consider a particular case of the (S1) rule where the two dots are of degree 2. The following derivation essentially consists in applying the bialgebra rule (B2) twice:
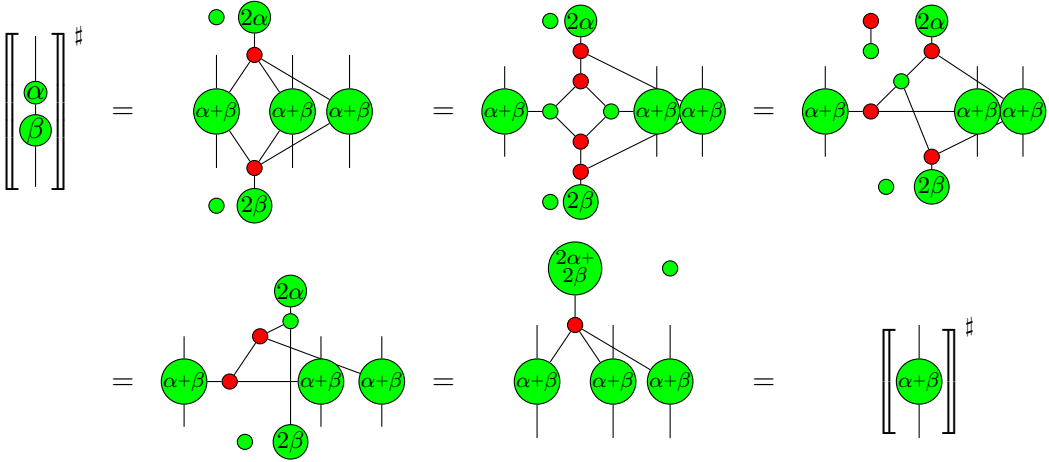
\begin{align*}\begin{array}{ccccccc}
\left\llbracket%
\beginpgfgraphicnamed{scalars//alpha-beta}
\begin{tikzpicture}
	\begin{pgfonlayer}{nodelayer}
		\node [style=gn] (0) at (0, 0.25) {$\alpha$};
		\node [style=none] (1) at (0, 1) {};
		\node [style=gn] (2) at (0, -0.25) {$\beta$};
		\node [style=none] (3) at (0, -1) {};
	\end{pgfonlayer}
	\begin{pgfonlayer}{edgelayer}
		\draw (1.center) to (0);
		\draw (0) to (2);
		\draw (2) to (3.center);
	\end{pgfonlayer}
\end{tikzpicture}}
\endpgfgraphicnamed\right\rrbracket^{\sharp}&=&%
\beginpgfgraphicnamed{scalars//interprealpha+betasc}
\InputIfFileExists{scalars//interprealpha+betasc.tikz}{}{\input{./figures/scalars//interprealpha+betasc.tikz}}
\endpgfgraphicnamed&=&%
\beginpgfgraphicnamed{scalars//interprealpha+beta2sc}
\InputIfFileExists{scalars//interprealpha+beta2sc.tikz}{}{\input{./figures/scalars//interprealpha+beta2sc.tikz}}
\endpgfgraphicnamed
&=&%
\beginpgfgraphicnamed{scalars//interprealpha+beta3sc}
\InputIfFileExists{scalars//interprealpha+beta3sc.tikz}{}{\input{./figures/scalars//interprealpha+beta3sc.tikz}}
\endpgfgraphicnamed\\
&&\\
&=&%
\beginpgfgraphicnamed{scalars//interprealpha+beta4sc}
\InputIfFileExists{scalars//interprealpha+beta4sc.tikz}{}{\input{./figures/scalars//interprealpha+beta4sc.tikz}}
\endpgfgraphicnamed&=&%
\beginpgfgraphicnamed{scalars//interprealpha+beta5sc}
\InputIfFileExists{scalars//interprealpha+beta5sc.tikz}{}{\input{./figures/scalars//interprealpha+beta5sc.tikz}}
\endpgfgraphicnamed&=&\left\llbracket%
\beginpgfgraphicnamed{scalars//alpha+beta}
\begin{tikzpicture}
	\begin{pgfonlayer}{nodelayer}
		\node [style=none] (0) at (0.5, -0.75) {};
		\node [style=none] (1) at (0.5, 0.75) {};
		\node [style=gn] (2) at (0.5, 0) {\scriptsize$\alpha{\splus}\beta$};
	\end{pgfonlayer}
	\begin{pgfonlayer}{edgelayer}
		\draw (1.center) to (2);
		\draw (2) to (0.center);
	\end{pgfonlayer}
\end{tikzpicture}}
\endpgfgraphicnamed\right\rrbracket^{\sharp}
 \end{array}
\end{align*}

\end{proof}

\begin{remark}
The interpretation $\den.$ can be naturally extended to an interpretation $\den._{k, \ell}$ which associates with every diagram $D:n\to m$ a diagram $\den D_{k,\ell}: k\times n \to k\times m$ which consists in $k$ copies of $D$ where the $k$ copies of each dot are connected by a ``gadget'' parameterized by an angle $\ell$ times larger than the angle of the original dot. Moreover $\den D_{k,\ell}$ has additional scalars, namely $k-1$ times $\droot$ per dot in $D$.
Notice that the interpretation $\den .$ used in this section is nothing but  $\den ._{3,2}$. The interpretation $\den._{k,\ell}$ is  sound  if and only if $k=1\bmod 2$ and $\ell =0\bmod 2$, indeed $(K1)$ forces $k$ to be odd while $(EU)$ and $(ZO)$ force $\ell = 0\bmod 2$. All the other rules are sound for any $k, \ell$. When $k=1$, $\den ._{1,\ell}$ is nothing but an interpretation which multiplies the angles by $\ell+1$, without changing the structure of the diagrams:  $\den . _{1,0}$ is the identity, while $\den._{1,-1}$ has been used to prove that the (EU) rule is necessary \cite{perdrixduncan2} and $\den._{1,-2}$ has been used to prove that the ZX-calculus is incomplete \cite{Vladimir}.
\end{remark}

\begin{proof}[Proof of Theorem \ref{thm:incomp}] In the following we prove that supplementarity can be derived in the ZX-calculus if and only if the involved angles are  multiples of $\pi/2$:

 {\emph{$$\left(ZX\vdash %
\beginpgfgraphicnamed{scalars//suppplenew}
\InputIfFileExists{scalars//suppplenew.tikz}{}{\input{./figures/scalars//suppplenew.tikz}}
\endpgfgraphicnamed\right) \qquad \Leftrightarrow \qquad \alpha = 0\bmod \frac \pi2 $$}}

\noindent[$\Leftarrow$] Since both diagrams of the supplementarity equation have the same standard interpretation $\frac 1{\sqrt 2}{1-e^{2i\alpha} \choose 0}$, by completeness of the $\frac \pi2$-fragment of the ZX-calculus, supplementarity can be derived when $\alpha$ is a multiple of $\frac \pi 2$.

\noindent[$\Rightarrow$]
Let $\alpha \in [0, 2\pi)$, and assume that supplementarity (\ref{eqn:supp}) can be derived in the ZX-calculus. Since $\den.$ is sound, the following equation must be derivable in the ZX-calculus:
\begin{eqnarray}\label{ct2}
\qquad\qquad\qquad{\left(%
\beginpgfgraphicnamed{scalars//011}
\InputIfFileExists{scalars//011.tikz}{}{\input{./figures/scalars//011.tikz}}
\endpgfgraphicnamed\right)} \circ \den{%
\beginpgfgraphicnamed{scalars//suppplenewleft}
\InputIfFileExists{scalars//suppplenewleft.tikz}{}{\input{./figures/scalars//suppplenewleft.tikz}}
\endpgfgraphicnamed} &=& \left(%
\beginpgfgraphicnamed{scalars//011}
\InputIfFileExists{scalars//011.tikz}{}{\input{./figures/scalars//011.tikz}}
\endpgfgraphicnamed\right) \circ {\den{%
\beginpgfgraphicnamed{scalars//suppplenewright}
\InputIfFileExists{scalars//suppplenewright.tikz}{}{\input{./figures/scalars//suppplenewright.tikz}}
\endpgfgraphicnamed}}
\end{eqnarray}
The LHS diagram is as follows. The details of the derivation are given in Lemma \ref{suplint011} in the appendix.
$$
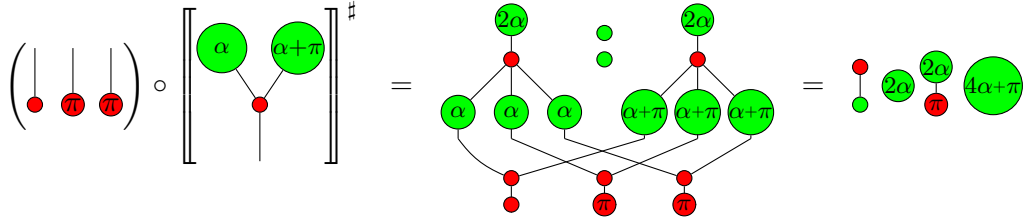
\left(%
\beginpgfgraphicnamed{scalars//011}
\InputIfFileExists{scalars//011.tikz}{}{\input{./figures/scalars//011.tikz}}
\endpgfgraphicnamed\right)\circ \left\llbracket %
\beginpgfgraphicnamed{scalars//suppplenewleft}
\InputIfFileExists{scalars//suppplenewleft.tikz}{}{\input{./figures/scalars//suppplenewleft.tikz}}
\endpgfgraphicnamed \right\rrbracket^{\sharp}~~=~~%
\beginpgfgraphicnamed{scalars//inptriplesupl011s}
\InputIfFileExists{scalars//inptriplesupl011s.tikz}{}{\input{./figures/scalars//inptriplesupl011s.tikz}}
\endpgfgraphicnamed~~=~~%
\beginpgfgraphicnamed{scalars//inptriplesupl011simpli}
\InputIfFileExists{scalars//inptriplesupl011simpli.tikz}{}{\input{./figures/scalars//inptriplesupl011simpli.tikz}}
\endpgfgraphicnamed
$$
The RHS diagram of Eq. \ref{ct2} is:
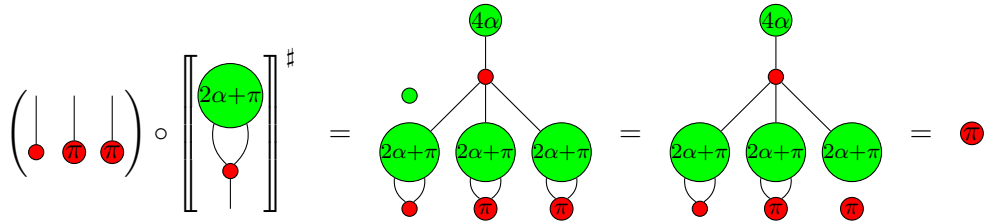
$$\left(%
\beginpgfgraphicnamed{scalars//011}
\InputIfFileExists{scalars//011.tikz}{}{\input{./figures/scalars//011.tikz}}
\endpgfgraphicnamed\right)\circ\left\llbracket%
\beginpgfgraphicnamed{scalars//suppplenewright}
\InputIfFileExists{scalars//suppplenewright.tikz}{}{\input{./figures/scalars//suppplenewright.tikz}}
\endpgfgraphicnamed\right\rrbracket^{\sharp}~~=~~%
\beginpgfgraphicnamed{scalars//supppleright011}
\InputIfFileExists{scalars//supppleright011.tikz}{}{\input{./figures/scalars//supppleright011.tikz}}
\endpgfgraphicnamed~~=~~%
\beginpgfgraphicnamed{scalars//supppleright011b}
\InputIfFileExists{scalars//supppleright011b.tikz}{}{\input{./figures/scalars//supppleright011b.tikz}}
\endpgfgraphicnamed~~=~~%
\beginpgfgraphicnamed{scalars//supppleright011simpli}
\begin{tikzpicture}
	\begin{pgfonlayer}{nodelayer}
		\node [style=rn] (0) at (0, -0) {$\pi$};
	\end{pgfonlayer}
\end{tikzpicture}}
\endpgfgraphicnamed$$
which is obtained first by applying the Hopf law and then thanks to the absorbing property of $\rpi$.

Thus, Eq. \ref{ct2} is equivalent to $%
\beginpgfgraphicnamed{scalars//inptriplesupl011simpli}
\InputIfFileExists{scalars//inptriplesupl011simpli.tikz}{}{\input{./figures/scalars//inptriplesupl011simpli.tikz}}
\endpgfgraphicnamed~=~%
\beginpgfgraphicnamed{scalars//supppleright011simpli}
\begin{tikzpicture}
	\begin{pgfonlayer}{nodelayer}
		\node [style=rn] (0) at (0, -0) {$\pi$};
	\end{pgfonlayer}
\end{tikzpicture}}
\endpgfgraphicnamed$
which can be simplified (see Lemma \ref{2a4apiispi} in appendix for details), leading to $%
\beginpgfgraphicnamed{scalars//2a4apiispileft}
\begin{tikzpicture}\footnotesize
	\begin{pgfonlayer}{nodelayer}
		\node [style=gn] (0) at (-1, -0) {$2\alpha$};
		\node [style=gn] (1) at (-0.25, -0) {$4\alpha{\splus}\pi$};
	\end{pgfonlayer}
\end{tikzpicture}}
\endpgfgraphicnamed~=~%
\beginpgfgraphicnamed{scalars//supppleright011simpli}
\begin{tikzpicture}
	\begin{pgfonlayer}{nodelayer}
		\node [style=rn] (0) at (0, -0) {$\pi$};
	\end{pgfonlayer}
\end{tikzpicture}}
\endpgfgraphicnamed$.

Finally, since $\denoteb{.}$ is sound, it implies $\denoteb{%
\beginpgfgraphicnamed{scalars//2a4apiispileft}
\begin{tikzpicture}\footnotesize
	\begin{pgfonlayer}{nodelayer}
		\node [style=gn] (0) at (-1, -0) {$2\alpha$};
		\node [style=gn] (1) at (-0.25, -0) {$4\alpha{\splus}\pi$};
	\end{pgfonlayer}
\end{tikzpicture}}
\endpgfgraphicnamed} = \denoteb{%
\beginpgfgraphicnamed{scalars//supppleright011simpli}
\begin{tikzpicture}
	\begin{pgfonlayer}{nodelayer}
		\node [style=rn] (0) at (0, -0) {$\pi$};
	\end{pgfonlayer}
\end{tikzpicture}}
\endpgfgraphicnamed}$, thus $(1+e^{2i\alpha})(1-e^{4i\alpha})=0$ which is equivalent to $\alpha = 0\bmod \frac \pi 2$. 
%
%
%
%
%
%
%
%
 \end{proof}

\section{Conclusion and further work}


In this paper, we have considered  supplementarity in the context of the ZX-calculus. We provide a purely graphical interpretation of supplementarity by means of antiphase twins. We have also proved that supplementarity can  be derived in the ZX-calculus if and only if the involved angles are  multiples of $\pi/2$. As a corollary, the $\frac \pi 4$-fragment of the ZX-calculus is not complete.

We propose to add the supplementarity rule (with arbitrary angles) to the set of rules of the ZX-calculus. Notice that even augmented with the supplementarity rule the ZX-calculus is still incomplete in general,  since the argument of \cite{Vladimir} still applies (the alternative interpretation which consists in multiplying the angles by an odd number, as the one used in \cite{Vladimir}, is sound with respect to the supplementarity rule).

We leave as an open question the completeness of the $\frac \pi 4$-fragment of the ZX-calculus augmented with the supplementarity rule, as well as any fragment which does not contain irrational multiples of $\pi$. Another perspective is to determine how the presence of the supplementarity rule impacts the other rules of the language. In particular, does supplementarity subsume any of the other rules of the ZX-calculus?

\section{Acknowledgements}

The authors would like to thank Miriam Backens, Bob Coecke, Ross Duncan, Emmanuel Jeandel, and Aleks Kissinger for valuable discussions. This work was partially supported by R\'egion Lorraine.

%
%

\section*{Appendix}

\begin{proof}[Proof of Lemma \ref{lem:sup}]
We prove the equivalences in the following order.

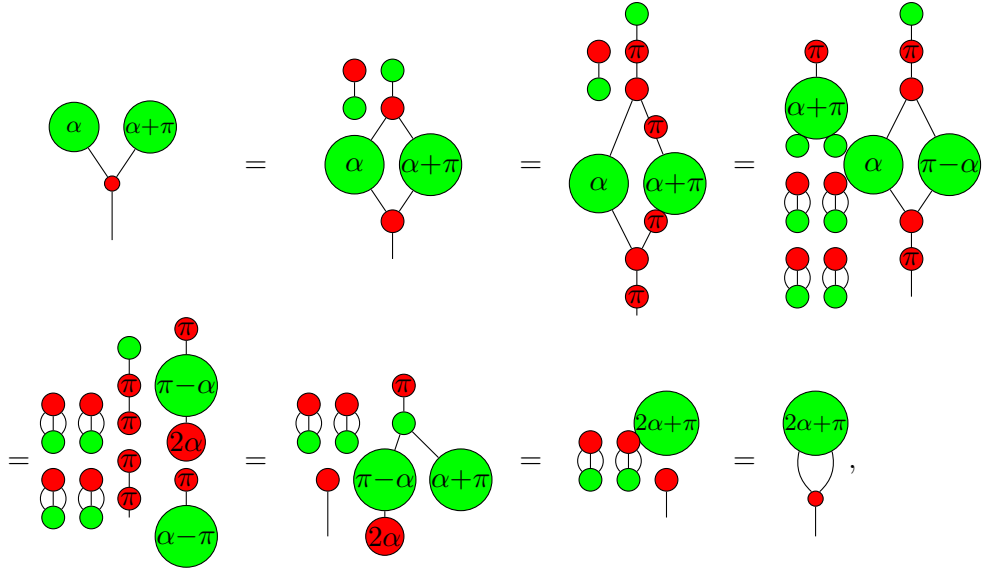
\begin{displaymath}
\begin{array}{ccccc}
\beginpgfgraphicnamed{scalars//suppplenew}
\InputIfFileExists{scalars//suppplenew.tikz}{}{\input{./figures/scalars//suppplenew.tikz}}
\endpgfgraphicnamed&\Rightarrow&%
\beginpgfgraphicnamed{scalars//supplepj0b}
\InputIfFileExists{scalars//supplepj0b.tikz}{}{\input{./figures/scalars//supplepj0b.tikz}}
\endpgfgraphicnamed&\Rightarrow& %
\beginpgfgraphicnamed{scalars//supplepj1}
\InputIfFileExists{scalars//supplepj1.tikz}{}{\input{./figures/scalars//supplepj1.tikz}}
\endpgfgraphicnamed\\
~\Rightarrow~ %
\beginpgfgraphicnamed{scalars//suppplenew}
\InputIfFileExists{scalars//suppplenew.tikz}{}{\input{./figures/scalars//suppplenew.tikz}}
\endpgfgraphicnamed.&&&&
\end{array}
\end{displaymath}

For
\begin{align*}
\beginpgfgraphicnamed{scalars//suppplenew}
\InputIfFileExists{scalars//suppplenew.tikz}{}{\input{./figures/scalars//suppplenew.tikz}}
\endpgfgraphicnamed~~~\Rightarrow~~~%
\beginpgfgraphicnamed{scalars//supplepj0b}
\InputIfFileExists{scalars//supplepj0b.tikz}{}{\input{./figures/scalars//supplepj0b.tikz}}
\endpgfgraphicnamed,
\end{align*}
we have
\begin{align*}
%
\beginpgfgraphicnamed{scalars//equivatob1}
\InputIfFileExists{scalars//equivatob1.tikz}{}{\input{./figures/scalars//equivatob1.tikz}}
\endpgfgraphicnamed=%
\beginpgfgraphicnamed{scalars//equivatob2}
\InputIfFileExists{scalars//equivatob2.tikz}{}{\input{./figures/scalars//equivatob2.tikz}}
\endpgfgraphicnamed=%
\beginpgfgraphicnamed{scalars//equivatob3}
\InputIfFileExists{scalars//equivatob3.tikz}{}{\input{./figures/scalars//equivatob3.tikz}}
\endpgfgraphicnamed
=%
\beginpgfgraphicnamed{scalars//equivatob4}
\InputIfFileExists{scalars//equivatob4.tikz}{}{\input{./figures/scalars//equivatob4.tikz}}
\endpgfgraphicnamed=%
\beginpgfgraphicnamed{scalars//equivatob5}
\InputIfFileExists{scalars//equivatob5.tikz}{}{\input{./figures/scalars//equivatob5.tikz}}
\endpgfgraphicnamed=%
\beginpgfgraphicnamed{scalars//equivatob6}
\InputIfFileExists{scalars//equivatob6.tikz}{}{\input{./figures/scalars//equivatob6.tikz}}
\endpgfgraphicnamed,
\end{align*}
where the bialgebra rule (B2), the Hopf law and the copy rule (B1) are used.

For
\begin{align*}
\beginpgfgraphicnamed{scalars//supplepj0b}
\InputIfFileExists{scalars//supplepj0b.tikz}{}{\input{./figures/scalars//supplepj0b.tikz}}
\endpgfgraphicnamed~~~\Rightarrow~~~%
\beginpgfgraphicnamed{scalars//supplepj1}
\InputIfFileExists{scalars//supplepj1.tikz}{}{\input{./figures/scalars//supplepj1.tikz}}
\endpgfgraphicnamed,
\end{align*}
we have
\begin{align*}
\beginpgfgraphicnamed{scalars//equivbtoc1}
\InputIfFileExists{scalars//equivbtoc1.tikz}{}{\input{./figures/scalars//equivbtoc1.tikz}}
\endpgfgraphicnamed=%
\beginpgfgraphicnamed{scalars//equivbtoc2}
\InputIfFileExists{scalars//equivbtoc2.tikz}{}{\input{./figures/scalars//equivbtoc2.tikz}}
\endpgfgraphicnamed=%
\beginpgfgraphicnamed{scalars//equivbtoc3}
\InputIfFileExists{scalars//equivbtoc3.tikz}{}{\input{./figures/scalars//equivbtoc3.tikz}}
\endpgfgraphicnamed=
\beginpgfgraphicnamed{scalars//equivbtoc4}
\InputIfFileExists{scalars//equivbtoc4.tikz}{}{\input{./figures/scalars//equivbtoc4.tikz}}
\endpgfgraphicnamed= %
\beginpgfgraphicnamed{scalars//equivbtoc5}
\InputIfFileExists{scalars//equivbtoc5.tikz}{}{\input{./figures/scalars//equivbtoc5.tikz}}
\endpgfgraphicnamed= %
\beginpgfgraphicnamed{scalars//equivbtoc6}
\InputIfFileExists{scalars//equivbtoc6.tikz}{}{\input{./figures/scalars//equivbtoc6.tikz}}
\endpgfgraphicnamed,
\end{align*}
where the $\pi$-commutation rule (K2), the inverse rule (IV) and the $\pi$-copy rule (K1) are used.

For
\begin{align*}
\beginpgfgraphicnamed{scalars//supplepj1}
\InputIfFileExists{scalars//supplepj1.tikz}{}{\input{./figures/scalars//supplepj1.tikz}}
\endpgfgraphicnamed~~~\Rightarrow~~~%
\beginpgfgraphicnamed{scalars//suppplenew}
\InputIfFileExists{scalars//suppplenew.tikz}{}{\input{./figures/scalars//suppplenew.tikz}}
\endpgfgraphicnamed,
\end{align*}
we have
\begin{displaymath}
\begin{array}{ccccccl}
\beginpgfgraphicnamed{scalars//suppplenewleft}
\InputIfFileExists{scalars//suppplenewleft.tikz}{}{\input{./figures/scalars//suppplenewleft.tikz}}
\endpgfgraphicnamed&=&%
\beginpgfgraphicnamed{scalars//equivctoa1}
\InputIfFileExists{scalars//equivctoa1.tikz}{}{\input{./figures/scalars//equivctoa1.tikz}}
\endpgfgraphicnamed&=&%
\beginpgfgraphicnamed{scalars//equivctoa2}
\InputIfFileExists{scalars//equivctoa2.tikz}{}{\input{./figures/scalars//equivctoa2.tikz}}
\endpgfgraphicnamed&=&%
\beginpgfgraphicnamed{scalars//equivctoa3}
\InputIfFileExists{scalars//equivctoa3.tikz}{}{\input{./figures/scalars//equivctoa3.tikz}}
\endpgfgraphicnamed\\
= %
\beginpgfgraphicnamed{scalars//equivctoa4}
\InputIfFileExists{scalars//equivctoa4.tikz}{}{\input{./figures/scalars//equivctoa4.tikz}}
\endpgfgraphicnamed&=& %
\beginpgfgraphicnamed{scalars//equivctoa5}
\InputIfFileExists{scalars//equivctoa5.tikz}{}{\input{./figures/scalars//equivctoa5.tikz}}
\endpgfgraphicnamed&=&%
\beginpgfgraphicnamed{scalars//equivctoa6}
\InputIfFileExists{scalars//equivctoa6.tikz}{}{\input{./figures/scalars//equivctoa6.tikz}}
\endpgfgraphicnamed&=&%
\beginpgfgraphicnamed{scalars//suppplenewright}
\InputIfFileExists{scalars//suppplenewright.tikz}{}{\input{./figures/scalars//suppplenewright.tikz}}
\endpgfgraphicnamed,
\end{array}
\end{displaymath}

where the copy rule (B1), $\pi$-commutation rule (K2),  inverse rule (IV), $\pi$-copy rule (K1), color change rule (H) and the Hopf law are the are used.

\end{proof}

\begin{lemma}\label{2hdecom}
In the scaled ZX-calculus,
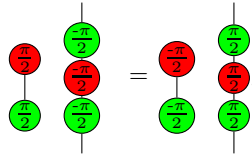
\begin{align*}
\beginpgfgraphicnamed{scalars//mphdecom}
\InputIfFileExists{scalars//mphdecom.tikz}{}{\input{./figures/scalars//mphdecom.tikz}}
\endpgfgraphicnamed.
\end{align*}
\end{lemma}

\begin{proof}
From the (H) rule, it is clear that %
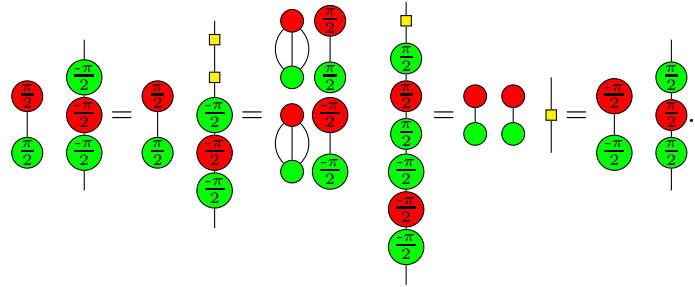
\beginpgfgraphicnamed{scalars//hsquareis1new}
\InputIfFileExists{scalars//hsquareis1new.tikz}{}{\input{./figures/scalars//hsquareis1new.tikz}}
\endpgfgraphicnamed. Therefore,
\begin{align*}
\beginpgfgraphicnamed{scalars//mphdecomprf1}
\InputIfFileExists{scalars//mphdecomprf1.tikz}{}{\input{./figures/scalars//mphdecomprf1.tikz}}
\endpgfgraphicnamed=%
\beginpgfgraphicnamed{scalars//mphdecomprf2new}
\InputIfFileExists{scalars//mphdecomprf2new.tikz}{}{\input{./figures/scalars//mphdecomprf2new.tikz}}
\endpgfgraphicnamed=%
\beginpgfgraphicnamed{scalars//mphdecomprf3new}
\InputIfFileExists{scalars//mphdecomprf3new.tikz}{}{\input{./figures/scalars//mphdecomprf3new.tikz}}
\endpgfgraphicnamed=%
\beginpgfgraphicnamed{scalars//mphdecomprf4new}
\InputIfFileExists{scalars//mphdecomprf4new.tikz}{}{\input{./figures/scalars//mphdecomprf4new.tikz}}
\endpgfgraphicnamed=%
\beginpgfgraphicnamed{scalars//mphdecomprf5}
\InputIfFileExists{scalars//mphdecomprf5.tikz}{}{\input{./figures/scalars//mphdecomprf5.tikz}}
\endpgfgraphicnamed.
\end{align*}
Note that we used Lemma 13 of \cite{Miriam3} for the third equality.

\end{proof}

\begin{proof}[Proof of Lemma \ref{lem:sound}]

To prove soundness of  $\llbracket \cdot \rrbracket^{\sharp}$, we need to verify that each rule in Figure \ref{figure1} as well as its colour-swapped and upside-down versions also hold under this interpretation. Since $\llbracket \cdot \rrbracket^{\sharp}$ enjoys `colour' and `upside-down' symmetries, we can only check the rules listed in Figure \ref{figure1} for soundness.
Firstly,  rules  $(S2),~(S3),~(B1),~(B2), ~(K1), ~(IV)$ and $(ZO)$ still hold under $\llbracket \cdot \rrbracket^{\sharp}$, since their interpretations are just triple copies of themselves. Secondly, the rule $(H)$ holds under $\llbracket \cdot \rrbracket^{\sharp}$, because the interpretation of red $\alpha$ is defined according to this rule. Finally, we check the rules $(S1),~(K2)$ and $(EU)$ in detail.

For $(S1)$, it suffices to prove
$$
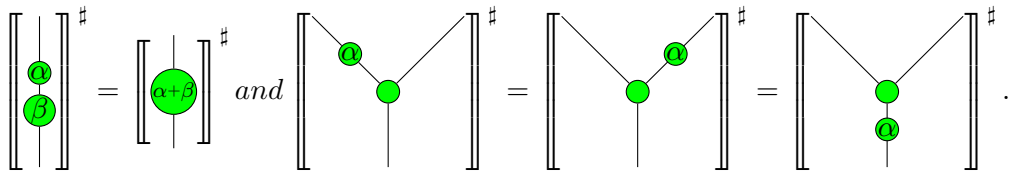
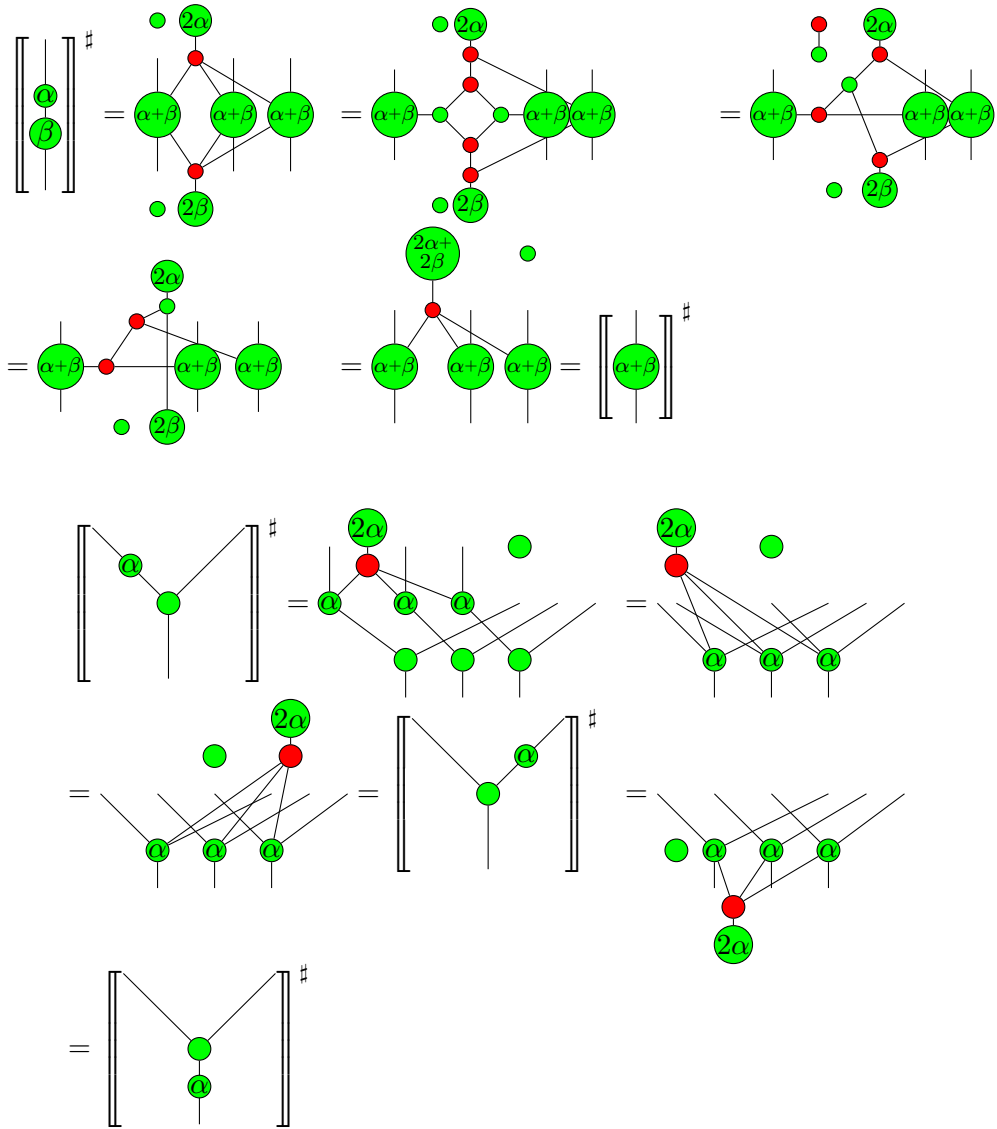
\left\llbracket%
\beginpgfgraphicnamed{scalars//alpha-beta}
\begin{tikzpicture}
	\begin{pgfonlayer}{nodelayer}
		\node [style=gn] (0) at (0, 0.25) {$\alpha$};
		\node [style=none] (1) at (0, 1) {};
		\node [style=gn] (2) at (0, -0.25) {$\beta$};
		\node [style=none] (3) at (0, -1) {};
	\end{pgfonlayer}
	\begin{pgfonlayer}{edgelayer}
		\draw (1.center) to (0);
		\draw (0) to (2);
		\draw (2) to (3.center);
	\end{pgfonlayer}
\end{tikzpicture}}
\endpgfgraphicnamed\right\rrbracket^{\sharp}=\left\llbracket%
\beginpgfgraphicnamed{scalars//alpha+beta}
\begin{tikzpicture}
	\begin{pgfonlayer}{nodelayer}
		\node [style=none] (0) at (0.5, -0.75) {};
		\node [style=none] (1) at (0.5, 0.75) {};
		\node [style=gn] (2) at (0.5, 0) {\scriptsize$\alpha{\splus}\beta$};
	\end{pgfonlayer}
	\begin{pgfonlayer}{edgelayer}
		\draw (1.center) to (2);
		\draw (2) to (0.center);
	\end{pgfonlayer}
\end{tikzpicture}}
\endpgfgraphicnamed\right\rrbracket^{\sharp}
and \left\llbracket%
\beginpgfgraphicnamed{scalars//leftalpha}
\InputIfFileExists{scalars//leftalpha.tikz}{}{\input{./figures/scalars//leftalpha.tikz}}
\endpgfgraphicnamed\right\rrbracket^{\sharp}=\left\llbracket%
\beginpgfgraphicnamed{scalars//rightalpha}
\InputIfFileExists{scalars//rightalpha.tikz}{}{\input{./figures/scalars//rightalpha.tikz}}
\endpgfgraphicnamed\right\rrbracket^{\sharp}
=\left\llbracket%
\beginpgfgraphicnamed{scalars//downalpha}
\InputIfFileExists{scalars//downalpha.tikz}{}{\input{./figures/scalars//downalpha.tikz}}
\endpgfgraphicnamed\right\rrbracket^{\sharp}.
$$

In fact,
\begin{align*}\begin{array}{lll}
\left\llbracket%
\beginpgfgraphicnamed{scalars//alpha-beta}
\begin{tikzpicture}
	\begin{pgfonlayer}{nodelayer}
		\node [style=gn] (0) at (0, 0.25) {$\alpha$};
		\node [style=none] (1) at (0, 1) {};
		\node [style=gn] (2) at (0, -0.25) {$\beta$};
		\node [style=none] (3) at (0, -1) {};
	\end{pgfonlayer}
	\begin{pgfonlayer}{edgelayer}
		\draw (1.center) to (0);
		\draw (0) to (2);
		\draw (2) to (3.center);
	\end{pgfonlayer}
\end{tikzpicture}}
\endpgfgraphicnamed\right\rrbracket^{\sharp}=%
\beginpgfgraphicnamed{scalars//interprealpha+betasc}
\InputIfFileExists{scalars//interprealpha+betasc.tikz}{}{\input{./figures/scalars//interprealpha+betasc.tikz}}
\endpgfgraphicnamed&=%
\beginpgfgraphicnamed{scalars//interprealpha+beta2sc}
\InputIfFileExists{scalars//interprealpha+beta2sc.tikz}{}{\input{./figures/scalars//interprealpha+beta2sc.tikz}}
\endpgfgraphicnamed
&=%
\beginpgfgraphicnamed{scalars//interprealpha+beta3sc}
\InputIfFileExists{scalars//interprealpha+beta3sc.tikz}{}{\input{./figures/scalars//interprealpha+beta3sc.tikz}}
\endpgfgraphicnamed\\
=%
\beginpgfgraphicnamed{scalars//interprealpha+beta4sc}
\InputIfFileExists{scalars//interprealpha+beta4sc.tikz}{}{\input{./figures/scalars//interprealpha+beta4sc.tikz}}
\endpgfgraphicnamed&=%
\beginpgfgraphicnamed{scalars//interprealpha+beta5sc}
\InputIfFileExists{scalars//interprealpha+beta5sc.tikz}{}{\input{./figures/scalars//interprealpha+beta5sc.tikz}}
\endpgfgraphicnamed=\left\llbracket%
\beginpgfgraphicnamed{scalars//alpha+beta}
\begin{tikzpicture}
	\begin{pgfonlayer}{nodelayer}
		\node [style=none] (0) at (0.5, -0.75) {};
		\node [style=none] (1) at (0.5, 0.75) {};
		\node [style=gn] (2) at (0.5, 0) {\scriptsize$\alpha{\splus}\beta$};
	\end{pgfonlayer}
	\begin{pgfonlayer}{edgelayer}
		\draw (1.center) to (2);
		\draw (2) to (0.center);
	\end{pgfonlayer}
\end{tikzpicture}}
\endpgfgraphicnamed\right\rrbracket^{\sharp}&
 \end{array}
\end{align*}

\begin{align*}\begin{array}{lll}
\left\llbracket%
\beginpgfgraphicnamed{scalars//leftalpha}
\InputIfFileExists{scalars//leftalpha.tikz}{}{\input{./figures/scalars//leftalpha.tikz}}
\endpgfgraphicnamed\right\rrbracket^{\sharp}=%
\beginpgfgraphicnamed{scalars//interprealphalbranchsc}
\InputIfFileExists{scalars//interprealphalbranchsc.tikz}{}{\input{./figures/scalars//interprealphalbranchsc.tikz}}
\endpgfgraphicnamed&=%
\beginpgfgraphicnamed{scalars//interprealphalbranch2sc}
\InputIfFileExists{scalars//interprealphalbranch2sc.tikz}{}{\input{./figures/scalars//interprealphalbranch2sc.tikz}}
\endpgfgraphicnamed&\\
=%
\beginpgfgraphicnamed{scalars//interprealphartbranchsc}
\InputIfFileExists{scalars//interprealphartbranchsc.tikz}{}{\input{./figures/scalars//interprealphartbranchsc.tikz}}
\endpgfgraphicnamed=\left\llbracket%
\beginpgfgraphicnamed{scalars//rightalpha}
\InputIfFileExists{scalars//rightalpha.tikz}{}{\input{./figures/scalars//rightalpha.tikz}}
\endpgfgraphicnamed\right\rrbracket^{\sharp}&=%
\beginpgfgraphicnamed{scalars//interprealphadownbranchsc}
\InputIfFileExists{scalars//interprealphadownbranchsc.tikz}{}{\input{./figures/scalars//interprealphadownbranchsc.tikz}}
\endpgfgraphicnamed&\\
=\left\llbracket%
\beginpgfgraphicnamed{scalars//downalpha}
\InputIfFileExists{scalars//downalpha.tikz}{}{\input{./figures/scalars//downalpha.tikz}}
\endpgfgraphicnamed\right\rrbracket^{\sharp}&&
 \end{array}
\end{align*}

For $(K2)$, on the one hand,

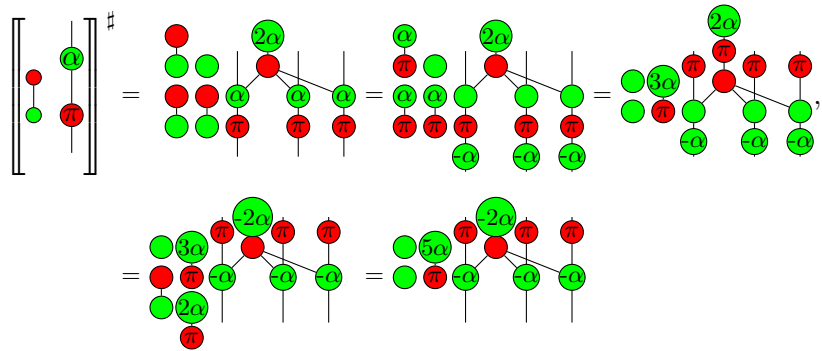
\begin{align*}
\left\llbracket%
\beginpgfgraphicnamed{scalars//alphapisc}
\InputIfFileExists{scalars//alphapisc.tikz}{}{\input{./figures/scalars//alphapisc.tikz}}
\endpgfgraphicnamed\right\rrbracket^{\sharp}%
\beginpgfgraphicnamed{scalars//interprepialphapfl}
\InputIfFileExists{scalars//interprepialphapfl.tikz}{}{\input{./figures/scalars//interprepialphapfl.tikz}}
\endpgfgraphicnamed,
\end{align*}

on the other hand,
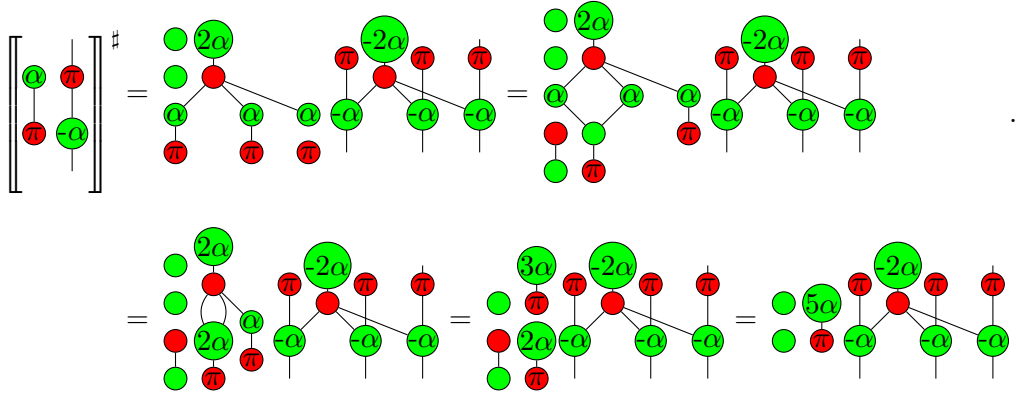
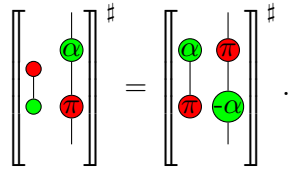
\begin{align*}
\left\llbracket%
\beginpgfgraphicnamed{scalars//piminusalphasc}
\InputIfFileExists{scalars//piminusalphasc.tikz}{}{\input{./figures/scalars//piminusalphasc.tikz}}
\endpgfgraphicnamed\right\rrbracket^{\sharp}%
\beginpgfgraphicnamed{scalars//interprepialphapfr}
\InputIfFileExists{scalars//interprepialphapfr.tikz}{}{\input{./figures/scalars//interprepialphapfr.tikz}}
\endpgfgraphicnamed.
\end{align*}

Therefore,
\begin{align*}
\left\llbracket%
\beginpgfgraphicnamed{scalars//alphapisc}
\InputIfFileExists{scalars//alphapisc.tikz}{}{\input{./figures/scalars//alphapisc.tikz}}
\endpgfgraphicnamed\right\rrbracket^{\sharp}=\left\llbracket%
\beginpgfgraphicnamed{scalars//piminusalphasc}
\InputIfFileExists{scalars//piminusalphasc.tikz}{}{\input{./figures/scalars//piminusalphasc.tikz}}
\endpgfgraphicnamed\right\rrbracket^{\sharp}.
\end{align*}

For $(EU)$, Firstly we have
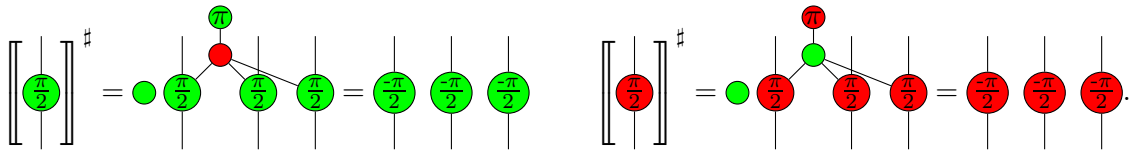
\begin{align*}
\left\llbracket %
\beginpgfgraphicnamed{scalars//alphapiby2}
\begin{tikzpicture}
	\begin{pgfonlayer}{nodelayer}
		\node [style=none] (0) at (0.5, -0.75) {};
		\node [style=gn] (1) at (0.5, 0) {$\frac{\pi}{2}$};
		\node [style=none] (2) at (0.5, 0.75) {};
	\end{pgfonlayer}
	\begin{pgfonlayer}{edgelayer}
		\draw (2.center) to (1);
		\draw (1) to (0.center);
	\end{pgfonlayer}
\end{tikzpicture}}
\endpgfgraphicnamed\right\rrbracket^{\sharp}=%
\beginpgfgraphicnamed{scalars//interprepiby2sc}
\InputIfFileExists{scalars//interprepiby2sc.tikz}{}{\input{./figures/scalars//interprepiby2sc.tikz}}
\endpgfgraphicnamed=%
\beginpgfgraphicnamed{scalars//3piby2}
\InputIfFileExists{scalars//3piby2.tikz}{}{\input{./figures/scalars//3piby2.tikz}}
\endpgfgraphicnamed\quad\quad
\left\llbracket %
\beginpgfgraphicnamed{scalars//alphapiby2red}
\begin{tikzpicture}
	\begin{pgfonlayer}{nodelayer}
		\node [style=none] (0) at (0.5, 0.75) {};
		\node [style=none] (1) at (0.5, -0.75) {};
		\node [style=rn] (2) at (0.5, 0) {$\frac{\pi}{2}$};
	\end{pgfonlayer}
	\begin{pgfonlayer}{edgelayer}
		\draw (0.center) to (2);
		\draw (2) to (1.center);
	\end{pgfonlayer}
\end{tikzpicture}}
\endpgfgraphicnamed\right\rrbracket^{\sharp}=%
\beginpgfgraphicnamed{scalars//interprepiby2redsc}
\InputIfFileExists{scalars//interprepiby2redsc.tikz}{}{\input{./figures/scalars//interprepiby2redsc.tikz}}
\endpgfgraphicnamed=%
\beginpgfgraphicnamed{scalars//3piby2red}
\InputIfFileExists{scalars//3piby2red.tikz}{}{\input{./figures/scalars//3piby2red.tikz}}
\endpgfgraphicnamed.
\end{align*}

Then it follows from Lemma \ref{2hdecom} that
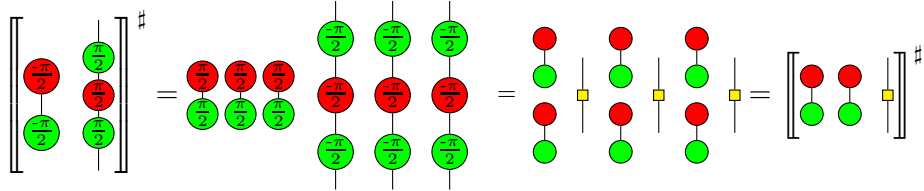
\begin{align*}
\left\llbracket %
\beginpgfgraphicnamed{scalars//grgpiby2sc}
\InputIfFileExists{scalars//grgpiby2sc.tikz}{}{\input{./figures/scalars//grgpiby2sc.tikz}}
\endpgfgraphicnamed\right\rrbracket^{\sharp}=%
\beginpgfgraphicnamed{scalars//eusoundpf}
\InputIfFileExists{scalars//eusoundpf.tikz}{}{\input{./figures/scalars//eusoundpf.tikz}}
\endpgfgraphicnamed
=\left\llbracket %
\beginpgfgraphicnamed{scalars//Had2sc}
\InputIfFileExists{scalars//Had2sc.tikz}{}{\input{./figures/scalars//Had2sc.tikz}}
\endpgfgraphicnamed\right\rrbracket^{\sharp}
\end{align*}

To sum up, the interpretation $\llbracket \cdot \rrbracket^{\sharp}$ is sound.

\end{proof}

\begin{lemma}\label{suplint011}
In the ZX-calculus,
\begin{align*}
\hspace{1.5cm} %
\beginpgfgraphicnamed{scalars//011}
\InputIfFileExists{scalars//011.tikz}{}{\input{./figures/scalars//011.tikz}}
\endpgfgraphicnamed\circ\left\llbracket %
\beginpgfgraphicnamed{scalars//suppplenewleft}
\InputIfFileExists{scalars//suppplenewleft.tikz}{}{\input{./figures/scalars//suppplenewleft.tikz}}
\endpgfgraphicnamed\right\rrbracket^{\sharp}=%
\beginpgfgraphicnamed{scalars//inptriplesupl011simpli}
\InputIfFileExists{scalars//inptriplesupl011simpli.tikz}{}{\input{./figures/scalars//inptriplesupl011simpli.tikz}}
\endpgfgraphicnamed.
\end{align*}
\end{lemma}

\begin{proof}
\[
\begin{array}{ll}
\beginpgfgraphicnamed{scalars//011}
\InputIfFileExists{scalars//011.tikz}{}{\input{./figures/scalars//011.tikz}}
\endpgfgraphicnamed\circ\left\llbracket %
\beginpgfgraphicnamed{scalars//suppplenewleft}
\InputIfFileExists{scalars//suppplenewleft.tikz}{}{\input{./figures/scalars//suppplenewleft.tikz}}
\endpgfgraphicnamed\right\rrbracket^{\sharp}=%
\beginpgfgraphicnamed{scalars//inptriplesupl011}
\InputIfFileExists{scalars//inptriplesupl011.tikz}{}{\input{./figures/scalars//inptriplesupl011.tikz}}
\endpgfgraphicnamed&\vspace{0.5cm}\\
=%
\beginpgfgraphicnamed{scalars//suplint011prf2}
\InputIfFileExists{scalars//suplint011prf2.tikz}{}{\input{./figures/scalars//suplint011prf2.tikz}}
\endpgfgraphicnamed=%
\beginpgfgraphicnamed{scalars//suplint011prf3}
\InputIfFileExists{scalars//suplint011prf3.tikz}{}{\input{./figures/scalars//suplint011prf3.tikz}}
\endpgfgraphicnamed&\vspace{0.5cm}\\
=%
\beginpgfgraphicnamed{scalars//suplint011prf4}
\InputIfFileExists{scalars//suplint011prf4.tikz}{}{\input{./figures/scalars//suplint011prf4.tikz}}
\endpgfgraphicnamed=%
\beginpgfgraphicnamed{scalars//suplint011prf5}
\InputIfFileExists{scalars//suplint011prf5.tikz}{}{\input{./figures/scalars//suplint011prf5.tikz}}
\endpgfgraphicnamed=%
\beginpgfgraphicnamed{scalars//suplint011prf6}
\InputIfFileExists{scalars//suplint011prf6.tikz}{}{\input{./figures/scalars//suplint011prf6.tikz}}
\endpgfgraphicnamed&\vspace{0.5cm}\\
=%
\beginpgfgraphicnamed{scalars//suplint011prf7}
\InputIfFileExists{scalars//suplint011prf7.tikz}{}{\input{./figures/scalars//suplint011prf7.tikz}}
\endpgfgraphicnamed=%
\beginpgfgraphicnamed{scalars//inptriplesupl011simpli}
\InputIfFileExists{scalars//inptriplesupl011simpli.tikz}{}{\input{./figures/scalars//inptriplesupl011simpli.tikz}}
\endpgfgraphicnamed.
\end{array}
\]
Here the spider rule (S1), $\pi$-copy rule (K1), $\pi$-commutation rule (K2) and the Hopf law are used.
\end{proof}

\begin{lemma}\label{2a4apiispi}
In the ZX-calculus,
\begin{align*}
\beginpgfgraphicnamed{scalars//inptriplesupl011simpli}
\InputIfFileExists{scalars//inptriplesupl011simpli.tikz}{}{\input{./figures/scalars//inptriplesupl011simpli.tikz}}
\endpgfgraphicnamed=%
\beginpgfgraphicnamed{scalars//supppleright011simpli}
\begin{tikzpicture}
	\begin{pgfonlayer}{nodelayer}
		\node [style=rn] (0) at (0, -0) {$\pi$};
	\end{pgfonlayer}
\end{tikzpicture}}
\endpgfgraphicnamed\quad\quad  \Rightarrow \quad\quad    %
\beginpgfgraphicnamed{scalars//2a4apiispi}
\begin{tikzpicture}
	\begin{pgfonlayer}{nodelayer}
		\node [style=gn] (0) at (-1, -0) {$2\alpha$};
		\node [style=gn] (1) at (-0.25, -0) {\scriptsize $4\alpha{+}\pi$};
		\node [style=rn] (2) at (1, -0) {$\pi$};
		\node [style=none] (3) at (0.5, -0) {$=$};
	\end{pgfonlayer}
\end{tikzpicture}}
\endpgfgraphicnamed.
\end{align*}
\end{lemma}

\begin{proof}
\begin{align*}
\beginpgfgraphicnamed{scalars//2a4apiispiproof}
\InputIfFileExists{scalars//2a4apiispiproof.tikz}{}{\input{./figures/scalars//2a4apiispiproof.tikz}}
\endpgfgraphicnamed.
\end{align*}
Here we used the $\pi$-copy rule (K1), inverse rule (IV) and the fact that $%
\beginpgfgraphicnamed{scalars//supppleright011simpli}
\begin{tikzpicture}
	\begin{pgfonlayer}{nodelayer}
		\node [style=rn] (0) at (0, -0) {$\pi$};
	\end{pgfonlayer}
\end{tikzpicture}}
\endpgfgraphicnamed$ can absorb any scalar.

\end{proof}

\ed


%
%
%
%
%
%
%


corrolaty

ZX is incomplete.

%
%
%
%
%
%
%
%
%
%
%
%
%
%
%
%
%
%
%
%

\subsection{Incompleteness result}

\begin{theorem}[Incompleteness]\label{th1} The following equation cannot be derived in the ZX calculus
\begin{equation}\label{ct}
 \hspace{4.5cm}%
\beginpgfgraphicnamed{scalars//countexamplescalar}
\InputIfFileExists{scalars//countexamplescalar.tikz}{}{\input{./figures/scalars//countexamplescalar.tikz}}
\endpgfgraphicnamed
\end{equation}
\end{theorem}

\noindent \begin{proof}
We prove this theorem by contradiction. Assume that equation (\ref{ct}) can be derived in the ZX calculus. Since both  $\llbracket \cdot \rrbracket$ and $\llbracket \cdot \rrbracket^{\sharp}$ provide  sound models of the calculus, there must hold the following equation:
\begin{equation}\label{ct2}
\left\llbracket \left\llbracket %
\beginpgfgraphicnamed{scalars//dmultiplysc}
\InputIfFileExists{scalars//dmultiplysc.tikz}{}{\input{./figures/scalars//dmultiplysc.tikz}}
\endpgfgraphicnamed\right\rrbracket^{\sharp}\right\rrbracket=\left\llbracket\left\llbracket%
\beginpgfgraphicnamed{scalars//xdotsc}
\begin{tikzpicture}
	\begin{pgfonlayer}{nodelayer}
		\node [style=none] (0) at (0, -0.5) {};
		\node [style=gn] (1) at (-0.5, -0) {$\frac{-\pi}{2}$};
		\node [style=rn] (2) at (0, 0.25) {};
	\end{pgfonlayer}
	\begin{pgfonlayer}{edgelayer}
		\draw (2) to (0.center);
	\end{pgfonlayer}
\end{tikzpicture}}
\endpgfgraphicnamed\right\rrbracket^{\sharp}\right\rrbracket
\end{equation}
By direct calculation, we have

\[
\begin{array}{ll}
\left\llbracket \left\llbracket %
\beginpgfgraphicnamed{scalars//dmultiplysc}
\InputIfFileExists{scalars//dmultiplysc.tikz}{}{\input{./figures/scalars//dmultiplysc.tikz}}
\endpgfgraphicnamed\right\rrbracket^{\sharp}\right\rrbracket=&\left\llbracket %
\beginpgfgraphicnamed{scalars//inptriplesc}
\InputIfFileExists{scalars//inptriplesc.tikz}{}{\input{./figures/scalars//inptriplesc.tikz}}
\endpgfgraphicnamed\right\rrbracket\\
&=4\sqrt{2}(-1+i)(\ket{000}+2\ket{011}+2\ket{101}+2\ket{110})
\end{array}
\]

\begin{equation*}
\left\llbracket\left\llbracket %
\beginpgfgraphicnamed{scalars//xdotsc}
\begin{tikzpicture}
	\begin{pgfonlayer}{nodelayer}
		\node [style=none] (0) at (0, -0.5) {};
		\node [style=gn] (1) at (-0.5, -0) {$\frac{-\pi}{2}$};
		\node [style=rn] (2) at (0, 0.25) {};
	\end{pgfonlayer}
	\begin{pgfonlayer}{edgelayer}
		\draw (2) to (0.center);
	\end{pgfonlayer}
\end{tikzpicture}}
\endpgfgraphicnamed\right\rrbracket^{\sharp}\right\rrbracket=\left\llbracket%
\beginpgfgraphicnamed{scalars//inpxdotsc}
\InputIfFileExists{scalars//inpxdotsc.tikz}{}{\input{./figures/scalars//inpxdotsc.tikz}}
\endpgfgraphicnamed\right\rrbracket
=2(-1+i)\ket{000}.
\end{equation*}
Clearly, equation (\ref{ct2}) does not hold. Therefore, equation (\ref{ct}) cannot be derived in the ZX calculus.
\end{proof}

\begin{corollary}\label{cr1}
The ZX Calculus is incomplete for Clifford + $T$  quantum mechanics.
\end{corollary}
\begin{proof}
On the one hand, we have
\begin{equation*}
\left\llbracket %
\beginpgfgraphicnamed{scalars//dmultiplysc}
\InputIfFileExists{scalars//dmultiplysc.tikz}{}{\input{./figures/scalars//dmultiplysc.tikz}}
\endpgfgraphicnamed\right\rrbracket=\left\llbracket%
\beginpgfgraphicnamed{scalars//xdotsc}
\begin{tikzpicture}
	\begin{pgfonlayer}{nodelayer}
		\node [style=none] (0) at (0, -0.5) {};
		\node [style=gn] (1) at (-0.5, -0) {$\frac{-\pi}{2}$};
		\node [style=rn] (2) at (0, 0.25) {};
	\end{pgfonlayer}
	\begin{pgfonlayer}{edgelayer}
		\draw (2) to (0.center);
	\end{pgfonlayer}
\end{tikzpicture}}
\endpgfgraphicnamed\right\rrbracket=\sqrt{2}(1-i)\ket{0}.
\end{equation*}
On the other hand, it follows from Theorem \ref{th1} that equation (\ref{ct}) cannot be derived in the ZX calculus. This means the ZX Calculus is incomplete for Clifford + $T$  quantum mechanics.

\end{proof}

If we use the normalization interpretation, than we get the proof of incompleteness in usual sense.

\section{Removal of zero-rule}

Finally, we augment the ZX-calculus+supplementarity with the empty rule. We prove step by step that the zero rule can be derived from the supplementarity rule and the empty rule.

\begin{lemma}\label{zerorule1}
\begin{equation}\label{zerorule1eq}
\beginpgfgraphicnamed{scalars//zerorule1}
\InputIfFileExists{scalars//zerorule1.tikz}{}{\input{./figures/scalars//zerorule1.tikz}}
\endpgfgraphicnamed.
\end{equation}
\end{lemma}

\begin{proof}
Let $\alpha=0$ in the supplementary rule, we have
\begin{align*}
\beginpgfgraphicnamed{scalars//0insupple}
\InputIfFileExists{scalars//0insupple.tikz}{}{\input{./figures/scalars//0insupple.tikz}}
\endpgfgraphicnamed.
\end{align*}
Applying the copy rule (B1) to this, we obtain
\begin{equation}\label{insup}
\beginpgfgraphicnamed{scalars//0insupple2}
\InputIfFileExists{scalars//0insupple2.tikz}{}{\input{./figures/scalars//0insupple2.tikz}}
\endpgfgraphicnamed.
\end{equation}
Composing with $%
\beginpgfgraphicnamed{scalars//redcopy}
\begin{tikzpicture}
	\begin{pgfonlayer}{nodelayer}
		\node [style=none] (0) at (-0.25, 0.25) {};
		\node [style=none] (1) at (0.25, 0.25) {};
		\node [style=rn] (2) at (0, -0) {};
		\node [style=none] (3) at (0, -0.25) {};
	\end{pgfonlayer}
	\begin{pgfonlayer}{edgelayer}
		\draw (0.center) to (2);
		\draw (1.center) to (2);
		\draw (2) to (3.center);
	\end{pgfonlayer}
\end{tikzpicture}}
\endpgfgraphicnamed$ on both sides of (\ref{insup}), it follows that
\begin{align*}
\beginpgfgraphicnamed{scalars//0insupple3}
\InputIfFileExists{scalars//0insupple3.tikz}{}{\input{./figures/scalars//0insupple3.tikz}}
\endpgfgraphicnamed.
\end{align*}
That is,
\begin{align*}
\beginpgfgraphicnamed{scalars//zerorule1}
\InputIfFileExists{scalars//zerorule1.tikz}{}{\input{./figures/scalars//zerorule1.tikz}}
\endpgfgraphicnamed.
\end{align*}
\end{proof}

\begin{corollary}\label{dotabsorb}

\begin{align*}
\beginpgfgraphicnamed{scalars//dotabsorb}
\begin{tikzpicture}
	\begin{pgfonlayer}{nodelayer}
		\node [style=gn] (0) at (1, -0) {};
		\node [style=gn] (1) at (0.5, 0) {$\pi$};
		\node [style=gn] (2) at (2, -0) {$\pi$};
		\node [style=none] (3) at (1.5, -0) {$=$};
	\end{pgfonlayer}
\end{tikzpicture}}
\endpgfgraphicnamed.
\end{align*}
\end{corollary}

\begin{proof}
On both sides of (\ref{zerorule1eq}), composing with $%
\beginpgfgraphicnamed{scalars//gndown}
\begin{tikzpicture}
	\begin{pgfonlayer}{nodelayer}
		\node [style=none] (0) at (0, -0.25) {};
		\node [style=gn] (1) at (0, 0.25) {};
	\end{pgfonlayer}
	\begin{pgfonlayer}{edgelayer}
		\draw (1) to (0.center);
	\end{pgfonlayer}
\end{tikzpicture}}
\endpgfgraphicnamed$ at the top and $%
\beginpgfgraphicnamed{scalars//gnup}
\begin{tikzpicture}
	\begin{pgfonlayer}{nodelayer}
		\node [style=none] (0) at (0, 0.25) {};
		\node [style=gn] (1) at (0, -0.25) {};
	\end{pgfonlayer}
	\begin{pgfonlayer}{edgelayer}
		\draw (1) to (0.center);
	\end{pgfonlayer}
\end{tikzpicture}}
\endpgfgraphicnamed$ at the bottom, one has
\begin{align*}
\beginpgfgraphicnamed{scalars//dotabsorbpf1}
\InputIfFileExists{scalars//dotabsorbpf1.tikz}{}{\input{./figures/scalars//dotabsorbpf1.tikz}}
\endpgfgraphicnamed.
\end{align*}
Since $%
\beginpgfgraphicnamed{scalars//2gns}
\begin{tikzpicture}
	\begin{pgfonlayer}{nodelayer}
		\node [style=gn] (0) at (0, -0.25) {};
		\node [style=gn] (1) at (0, 0.25) {};
	\end{pgfonlayer}
	\begin{pgfonlayer}{edgelayer}
		\draw (1) to (0);
	\end{pgfonlayer}
\end{tikzpicture}}
\endpgfgraphicnamed$ has an inverse, we have
\begin{align*}
\beginpgfgraphicnamed{scalars//dotabsorb}
\begin{tikzpicture}
	\begin{pgfonlayer}{nodelayer}
		\node [style=gn] (0) at (1, -0) {};
		\node [style=gn] (1) at (0.5, 0) {$\pi$};
		\node [style=gn] (2) at (2, -0) {$\pi$};
		\node [style=none] (3) at (1.5, -0) {$=$};
	\end{pgfonlayer}
\end{tikzpicture}}
\endpgfgraphicnamed.
\end{align*}

\end{proof}

\begin{lemma}\label{piabsorbalpha}
\begin{align*}
\beginpgfgraphicnamed{scalars//piabsorbalpha}
\begin{tikzpicture}
	\begin{pgfonlayer}{nodelayer}
		\node [style=gn] (0) at (2, 0) {$\pi$};
		\node [style=gn] (1) at (0.5, 0) {$\pi$};
		\node [style=none] (2) at (1.5, 0) {$=$};
		\node [style=gn] (3) at (1, 0) {$\alpha$};
	\end{pgfonlayer}
\end{tikzpicture}}
\endpgfgraphicnamed.
\end{align*}
\end{lemma}

\begin{proof}
Adding the scalar $%
\beginpgfgraphicnamed{scalars//1line}
\begin{tikzpicture}[scale=0.55]
	\begin{pgfonlayer}{nodelayer}
		\node [style=srn] (0) at (0, 0.5) {};
		\node [style=sgn] (1) at (0, -0) {};
	\end{pgfonlayer}
	\begin{pgfonlayer}{edgelayer}
		\draw (0) to (1);
	\end{pgfonlayer}
\end{tikzpicture}}
\endpgfgraphicnamed$ to each side of the equation (\ref{zerorule1eq}), by equation (\ref{insup}) we have
\begin{equation}\label{piabsorbalphapf1}
\beginpgfgraphicnamed{scalars//piabsorbalphapf1}
\InputIfFileExists{scalars//piabsorbalphapf1.tikz}{}{\input{./figures/scalars//piabsorbalphapf1.tikz}}
\endpgfgraphicnamed.
\end{equation}
On both sides of the second equality of (\ref{piabsorbalphapf1}), composing with $%
\beginpgfgraphicnamed{scalars//gndown}
\begin{tikzpicture}
	\begin{pgfonlayer}{nodelayer}
		\node [style=none] (0) at (0, -0.25) {};
		\node [style=gn] (1) at (0, 0.25) {};
	\end{pgfonlayer}
	\begin{pgfonlayer}{edgelayer}
		\draw (1) to (0.center);
	\end{pgfonlayer}
\end{tikzpicture}}
\endpgfgraphicnamed$ at the top and $%
\beginpgfgraphicnamed{scalars//alphaup}
\begin{tikzpicture}
	\begin{pgfonlayer}{nodelayer}
		\node [style=none] (0) at (0.5, 0.25) {};
		\node [style=gn] (1) at (0.5, -0.25) {$\alpha$};
	\end{pgfonlayer}
	\begin{pgfonlayer}{edgelayer}
		\draw (1) to (0.center);
	\end{pgfonlayer}
\end{tikzpicture}}
\endpgfgraphicnamed$ at the bottom,
we have
\begin{align*}
\beginpgfgraphicnamed{scalars//piabsorbalphapf2}
\InputIfFileExists{scalars//piabsorbalphapf2.tikz}{}{\input{./figures/scalars//piabsorbalphapf2.tikz}}
\endpgfgraphicnamed.
\end{align*}
By lemma \ref{alphadelete}, that is,
\begin{align*}
\beginpgfgraphicnamed{scalars//piabsorbalphapf3}
\InputIfFileExists{scalars//piabsorbalphapf3.tikz}{}{\input{./figures/scalars//piabsorbalphapf3.tikz}}
\endpgfgraphicnamed.
\end{align*}
Since $%
\beginpgfgraphicnamed{scalars//1line}
\begin{tikzpicture}[scale=0.55]
	\begin{pgfonlayer}{nodelayer}
		\node [style=srn] (0) at (0, 0.5) {};
		\node [style=sgn] (1) at (0, -0) {};
	\end{pgfonlayer}
	\begin{pgfonlayer}{edgelayer}
		\draw (0) to (1);
	\end{pgfonlayer}
\end{tikzpicture}}
\endpgfgraphicnamed$ has an inverse, by corollary \ref{dotabsorb}, we have
\begin{align*}
\beginpgfgraphicnamed{scalars//piabsorbalpha}
\begin{tikzpicture}
	\begin{pgfonlayer}{nodelayer}
		\node [style=gn] (0) at (2, 0) {$\pi$};
		\node [style=gn] (1) at (0.5, 0) {$\pi$};
		\node [style=none] (2) at (1.5, 0) {$=$};
		\node [style=gn] (3) at (1, 0) {$\alpha$};
	\end{pgfonlayer}
\end{tikzpicture}}
\endpgfgraphicnamed.
\end{align*}
\end{proof}

\begin{lemma}\label{piroot2}
\begin{align*}
\beginpgfgraphicnamed{scalars//piroot2}
\begin{tikzpicture}
	\begin{pgfonlayer}{nodelayer}
		\node [style=gn] (0) at (0, -0) {$\pi$};
		\node [style=rn] (1) at (0.5, 0.25) {};
		\node [style=gn] (2) at (0.5, -0.25) {};
		\node [style=none] (3) at (1, -0) {$=$};
		\node [style=gn] (4) at (1.5, -0) {$\pi$};
	\end{pgfonlayer}
	\begin{pgfonlayer}{edgelayer}
		\draw (1) to (2);
	\end{pgfonlayer}
\end{tikzpicture}}
\endpgfgraphicnamed.
\end{align*}
\end{lemma}

\begin{proof}
Using the empty rule, lemma \ref{zerorule1}, corollary \ref{dotabsorb}, and  lemma \ref{piabsorbalpha}, we have
\begin{align*}
\beginpgfgraphicnamed{scalars//piroot2pr0}
\begin{tikzpicture}
	\begin{pgfonlayer}{nodelayer}
		\node [style=gn] (0) at (0.5, 0) {$\pi$};
		\node [style=gn] (1) at (1, -0.25) {};
		\node [style=rn] (2) at (1, 0.25) {};
	\end{pgfonlayer}
	\begin{pgfonlayer}{edgelayer}
		\draw (2) to (1);
	\end{pgfonlayer}
\end{tikzpicture}}
\endpgfgraphicnamed=%
\beginpgfgraphicnamed{scalars//piroot2pr1}
\InputIfFileExists{scalars//piroot2pr1.tikz}{}{\input{./figures/scalars//piroot2pr1.tikz}}
\endpgfgraphicnamed=%
\beginpgfgraphicnamed{scalars//piroot2pr2}
\InputIfFileExists{scalars//piroot2pr2.tikz}{}{\input{./figures/scalars//piroot2pr2.tikz}}
\endpgfgraphicnamed=%
\beginpgfgraphicnamed{scalars//piroot2pr3}
\InputIfFileExists{scalars//piroot2pr3.tikz}{}{\input{./figures/scalars//piroot2pr3.tikz}}
\endpgfgraphicnamed
=%
\beginpgfgraphicnamed{scalars//piroot2pr4}
\begin{tikzpicture}
	\begin{pgfonlayer}{nodelayer}
		\node [style=gn] (0) at (1, -0) {};
		\node [style=gn] (1) at (0.5, 0) {$\pi$};
	\end{pgfonlayer}
\end{tikzpicture}}
\endpgfgraphicnamed=%
\beginpgfgraphicnamed{scalars//gnpi}
\begin{tikzpicture}
	\begin{pgfonlayer}{nodelayer}
		\node [style=gn] (0) at (0, -0) {$\pi$};
	\end{pgfonlayer}
\end{tikzpicture}}
\endpgfgraphicnamed.
\end{align*}
\end{proof}

\begin{theorem}\label{zerorule2}
\begin{align*}
\beginpgfgraphicnamed{scalars//zerorule2}
\InputIfFileExists{scalars//zerorule2.tikz}{}{\input{./figures/scalars//zerorule2.tikz}}
\endpgfgraphicnamed.
\end{align*}
\end{theorem}

\begin{proof}
Using lemma \ref{piroot2}, we have directly from  the second equality of (\ref{piabsorbalphapf1}) that
\begin{align*}
\beginpgfgraphicnamed{scalars//zerorule2}
\InputIfFileExists{scalars//zerorule2.tikz}{}{\input{./figures/scalars//zerorule2.tikz}}
\endpgfgraphicnamed.
\end{align*}

\end{proof}

{\bf Open Question:}


Is ZX-calculus + supplementarity + empty-rule complete for $\frac{\pi}{4}$ fragment QM?

\section{Conclusion and further work}
Clifford+T  quantum mechanics is a key fragment of pure qubit quantum mechanics: it has infinite operations on a finite number of qubits and allows any unitary operation to be approximated to arbitrary accuracy. As a powerful graphical language for quantum processes, it is important to know whether the ZX-calculus is complete for Clifford+T quantum mechanics. In this paper, we show that a graphical equation in the ZX-calculus for Clifford + T  quantum mechanics cannot be derived from all the given rewrite rules, while the standard interpretations of both sides of the equation are the same up to a global scalar.  Therefore, we proved the incompleteness of  ZX-calculus for Clifford+T  quantum mechanics. The proof relies on a new non-standard interpretation of the ZX calculus.

We propose to augment the ZX-calculus with the new axiom of supplementarity. The obvious next step is to prove the completeness or incompleteness of this new language for Clifford+T quantum mechanics.

\subparagraph*{Acknowledgements}

We want to thank Miriam Backens, Bob Coecke, Ross Duncan, Aleks Kissinger, Emmanuel Jeandel, and Vladimir Zamdzhiev for fruitful and stimulating discussions.

\appendix
\section{Proof of Lemma \ref{lem:sup}}
\emph{Proof of Lemma \ref{lem:sup}}.

We prove the equivalences in the following order.
\begin{displaymath}
\begin{array}{ccccc}
\beginpgfgraphicnamed{scalars//suppplenew}
\InputIfFileExists{scalars//suppplenew.tikz}{}{\input{./figures/scalars//suppplenew.tikz}}
\endpgfgraphicnamed
&\Rightarrow&
\beginpgfgraphicnamed{scalars//countgen2scnew}
\InputIfFileExists{scalars//countgen2scnew.tikz}{}{\input{./figures/scalars//countgen2scnew.tikz}}
\endpgfgraphicnamed
&\Rightarrow&%
\beginpgfgraphicnamed{scalars//supplepj1}
\InputIfFileExists{scalars//supplepj1.tikz}{}{\input{./figures/scalars//supplepj1.tikz}}
\endpgfgraphicnamed\\
~~~\Rightarrow~~~%
\beginpgfgraphicnamed{scalars//supplepj0}
\InputIfFileExists{scalars//supplepj0.tikz}{}{\input{./figures/scalars//supplepj0.tikz}}
\endpgfgraphicnamed&\Rightarrow& %
\beginpgfgraphicnamed{scalars//suppplenew}
\InputIfFileExists{scalars//suppplenew.tikz}{}{\input{./figures/scalars//suppplenew.tikz}}
\endpgfgraphicnamed&&
\end{array}
\end{displaymath}

For
\begin{align*}
\beginpgfgraphicnamed{scalars//suppplenew}
\InputIfFileExists{scalars//suppplenew.tikz}{}{\input{./figures/scalars//suppplenew.tikz}}
\endpgfgraphicnamed~~~\Rightarrow~~~%
\beginpgfgraphicnamed{scalars//countgen2scnew}
\InputIfFileExists{scalars//countgen2scnew.tikz}{}{\input{./figures/scalars//countgen2scnew.tikz}}
\endpgfgraphicnamed
\end{align*}
we have
\begin{align*}
%
\beginpgfgraphicnamed{scalars//equivproof1}
\InputIfFileExists{scalars//equivproof1.tikz}{}{\input{./figures/scalars//equivproof1.tikz}}
\endpgfgraphicnamed
\end{align*}
where the commutation rule is used.

For
\begin{align*}
\beginpgfgraphicnamed{scalars//countgen2scnew}
\InputIfFileExists{scalars//countgen2scnew.tikz}{}{\input{./figures/scalars//countgen2scnew.tikz}}
\endpgfgraphicnamed~~\Rightarrow ~~%
\beginpgfgraphicnamed{scalars//supplepj1}
\InputIfFileExists{scalars//supplepj1.tikz}{}{\input{./figures/scalars//supplepj1.tikz}}
\endpgfgraphicnamed
\end{align*}
we have
\begin{align*}
%
\beginpgfgraphicnamed{scalars//equivproof2}
\InputIfFileExists{scalars//equivproof2.tikz}{}{\input{./figures/scalars//equivproof2.tikz}}
\endpgfgraphicnamed
\end{align*}
where the bialgebra rule, the Hopf law, the $\pi$-copy rule, the inverse rule and the colour change rule are used.

For
\begin{align*}
\beginpgfgraphicnamed{scalars//supplepj1}
\InputIfFileExists{scalars//supplepj1.tikz}{}{\input{./figures/scalars//supplepj1.tikz}}
\endpgfgraphicnamed~~\Rightarrow ~~%
\beginpgfgraphicnamed{scalars//supplepj0}
\InputIfFileExists{scalars//supplepj0.tikz}{}{\input{./figures/scalars//supplepj0.tikz}}
\endpgfgraphicnamed
\end{align*}
we have
\begin{align*}
%
\beginpgfgraphicnamed{scalars//equivproof3}
\InputIfFileExists{scalars//equivproof3.tikz}{}{\input{./figures/scalars//equivproof3.tikz}}
\endpgfgraphicnamed
\end{align*}
where  the $\pi$-copy rule, the inverse rule and the colour change rule are used.

For
\begin{align*}
%
\beginpgfgraphicnamed{scalars//supplepj0}
\InputIfFileExists{scalars//supplepj0.tikz}{}{\input{./figures/scalars//supplepj0.tikz}}
\endpgfgraphicnamed~~\Rightarrow ~~ %
\beginpgfgraphicnamed{scalars//suppplenew}
\InputIfFileExists{scalars//suppplenew.tikz}{}{\input{./figures/scalars//suppplenew.tikz}}
\endpgfgraphicnamed
\end{align*}
we have
\begin{align*}
 %
\beginpgfgraphicnamed{scalars//equivproof4}
\InputIfFileExists{scalars//equivproof4.tikz}{}{\input{./figures/scalars//equivproof4.tikz}}
\endpgfgraphicnamed
\end{align*}
where  the $\pi$-copy rule, the inverse rule and  the Hopf law are used. \hfill $\Box$






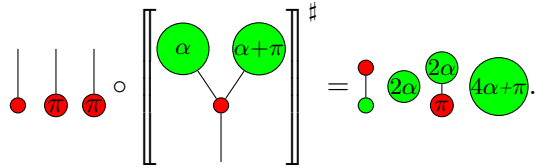
\end{document}